\documentclass[draftclsnofoot,12pt,onecolumn]{IEEEtran}
\usepackage {amsthm}
\usepackage {graphicx}
\usepackage {amsmath}
\usepackage {amssymb}
\usepackage {ulem}
\usepackage {cases}
\usepackage {latexsym}
\usepackage {multirow}
\usepackage{array}
\usepackage[dvips]{epsfig}
\usepackage{subfigure}
\theoremstyle{theorem}
\newtheorem{thm}{\textbf{Theorem}}
\newtheorem{lemma}{Lemma}
\newtheorem{prop}{\textbf{Proposition}}
\newtheorem{coll}{\textbf{Corollary}}

\theoremstyle{remark}
\newtheorem{rmk}{\textbf{Remark}}
\begin{document}
\title{Open, Closed, and Shared Access Femtocells in the Downlink}
\author{Han-Shin Jo,~%\IEEEmembership{Student Member,~IEEE,}
        Ping Xia,~%\IEEEmembership{Student Member,~IEEE,}
        and Jeffrey G. Andrews~%\IEEEmembership{Senior Member,~IEEE,}      % <-this % stops a space
\thanks{Han-Shin Jo, Ping Xia, and Jeffrey G. Andrews are with Wireless Networking and Communications Group, Department of Electrical and Computer Engineering, The University of Texas at Austin, Austin, TX78712-0240 USA
(e-mail: han-shin.jo@austin.utexas.edu, pxia@mail.utexas.edu and jandrews@ece.utexas.edu).}}
\markboth{} {Shell \MakeLowercase{\textit{et al.}}: Bare Demo of
IEEEtran.cls for Journals} \maketitle

\begin{abstract}
A fundamental choice in femtocell deployments is the set of users which are allowed to access each femtocell. Closed access restricts the set to specifically registered users, while open access allows any mobile subscriber to use any femtocell.  Which one is preferable depends strongly on the distance between the macrocell base station (MBS) and femtocell.  The main results of the paper are lemmas which provide expressions for the SINR distribution for various zones within a cell as a function of this MBS-femto distance.  The average sum throughput (or any other SINR-based metric) of home users and cellular users under open and closed access can be readily determined from these expressions.  We show that unlike in the uplink, the interests of home and cellular users are in conflict, with home users preferring closed access and cellular users preferring open access.  The conflict is most pronounced for femtocells near the cell edge, when there are many cellular users and fewer femtocells.  To mitigate this conflict, we propose a middle way which we term {\it shared} access in which femtocells allocate an adjustable number of time-slots between home and cellular users such that a specified minimum rate for each can be achieved.  The optimal such sharing fraction is derived.  Analysis shows that shared access achieves at least the overall throughput of open access while also satisfying rate requirements, while closed access fails for cellular users and open access fails for the home user.
\end{abstract}
\IEEEpeerreviewmaketitle

\section{Introduction}
Femtocells are small form-factor base stations that can be installed within an existing cellular network.  They can be installed either by an end-user or by the service provider and are distinguished from pico or microcells by their low cost and power and use of basic IP backhaul, and from WiFi by their use of cellular standards and licensed spectrum. Femtocells are a very promising and scalable method for meeting the ever-increasing demands for capacity and high-rate coverage. Since femtocells share spectrum with macrocell networks, managing cross-tier interference between femto- and macrocells is essential \cite{FemtoIntro1}-\cite{FemtoIntro3}. Furthermore, a basic question, particularly for end-user installed femtocells, is which users in the network should be allowed to use a given femtocell.

\subsection{Motivation and Related Work}
Cross-tier interference is highly dependent on this femtocell access decision.  {\it Closed access}, where only specified registered home users can communicate with the femtocell access point (FAP), appears attractive to the home user but can result in severe cross-tier interference from nearby cellular users in the uplink (see Fig. \ref{fig:Loudneighbor}) or to nearby cellular users in the downlink. To reduce this interference in closed access, previous studies have considered power control \cite{FemtoPower1}-\cite{FemtoPower5}, frequency assignment \cite{FemtoFreq1}-\cite{FemtoFreq3}, and a spectrum sensing approach \cite{FemtoCR1,FemtoCR2}. An alternative is to simply hand over cellular users that cause or experience strong interference to the femtocell. This is known as {\it open access}. Intuitively, this should increase the overall network capacity \cite{FemtoAccess1} at the possible expense of a given femtocell owner, who must now share his femtocell resources (time/frequency slots and backhaul) with an unpredictable number of cellular users.

The uplink performance of femtocell access schemes has been investigated in \cite{FemtoAccess2,FemtoAccess3}. The interrelationship between the traffic type, access policy, and performance of high-speed packet access (HSPA) was examined in a simulation-based study \cite{FemtoAccess2}. In \cite{FemtoAccess3}, an analytical framework was presented from open vs. closed access. Both studies suggested a hybrid access approach with an upper limit to the number of unregistered users to access the femtocell. We term this approach {\it shared access} since the femtocell is shared with cellular users, but within limits and hence not fully open. With respect to the uplink throughput of registered home users, open (or shared) access reduces interference by handing over the loud neighbor at the expense of FAP resource sharing. The tradeoff is such that open access is generally preferred for both home users and cellular users, since the interference reduction is so important \cite{FemtoAccess3}. Does the same tradeoff hold in the downlink?

It would seem that the tradeoff is different in the downlink since here the FAP is the loud neighbor (see Fig. \ref{fig:Loudneighbor}). Therefore, serving unregistered users with the FAP benefits them at the cost of FAP resources.  Meanwhile, there is at best a very small decrease in downlink interference to the home user.  The downlink capacity of open vs. closed access has been studied using simulations for HSPA femtocells \cite{FemtoAccess2,FemtoAccess4} and OFDMA femtocells \cite{FemtoAccess5}-\cite{FemtoAccess7}. These studies propose and analyze shared access methods with limits on the number of unregistered users \cite{FemtoAccess4} and frequency subchannels for them \cite{FemtoAccess5,FemtoAccess6}. Indeed, these works find that cellular user performance is improved with open (or shared) access at the cost of reduced home user performance. All these downlink simulations are for very specific scenarios, for example the throughput is averaged over all femtocells and a fixed number of femtocells and outdoor cellular users are considered.

%Femtocell access points (FAPs) installed by the end customer are distributed with randomness rather than regular pattern. Prior research in femtocell networks \cite{FemtoPower5},\cite{FemtoFreq3},\cite{FemtoUL} modeled the random spatial distribution of femtocells, as a Poisson Point Process (PPP) \cite{PPP1}-\cite{Baccelli}. Femtocell downlink analysis based on the stochastic geometry framework deals with effect of aggregated interference from randomly distributed femtocells, which provides more rigorous results as shown in many previous works \cite{FemtoPower5},\cite{FemtoFreq3}. Chandrasekhar {\it et al}. \cite{FemtoFreq3} derived the expected per-tier throughput in downlink where the macrocell and femtocells are frequency orthogonal, i.e., no cross-tier interference. However, co-channel femtocells are of great interest because it effectively utilizes frequency resources. Furthermore, the work deals with closed access only. Our work has extended this analysis to co-channel femtocells with open and closed access.

\subsection{Contributions and Main Insights}
Clearly, a more general and analytical approach is desirable. It should include key system parameters such as the distance between FAP and MBS, cell sizes, and the density of femtocells and cellular users.  It would be more realistic if the femtocell and cellular user positions were not fixed, but rather were modeled as a spatial random process (see \cite{PPP2,AndGan10} and references therein).  Ideally, a general statistical distribution of the SINR could be found for both closed and open access. Since metrics such as outage probability, error probability, and throughput follow directly from SINR, once the SINR distribution is known these metrics can be computed quite quickly and easily \cite{AndGan10}. Deriving such an SIR distribution (we neglect both thermal noise and interference from other cells) is the main contribution of the paper, and is used to draw a few conclusions about access strategies in the downlink.

First, we see that unlike the uplink \cite{FemtoAccess3}, the preferred access schemes for home and cellular users are incompatible, with home user preferring closed access. For femtocells with coverage area extending outside the home, i.e. far from the MBS, closed access provides higher sum throughput for home user and lower sum throughput for neighboring cellular users, when compared to open access. For example, of a cell edge femtocell, open access causes at least 20\% throughput loss to home user compared to closed access, while the neighboring cellular user experiences outage for typical data service (less than 15 kbps for 5 MHz bandwidth) in closed access.
When the cellular user density is high (and/or femtocell density is low), the performance difference between open and closed access increases, since cellular users are increasingly impinged upon by the femtocells and vice versa. Furthermore, we observe that the open access femtocells far from MBS reduce the macrocell load, thereby open access rather than closed access offers higher throughput for a few home users (in its femtocell coverage area smaller than home area) located near and connected to the MBS. Nevertheless, most home users in cell site accessing FAPs are still reluctant to use their femtocells in open access.

Since neither open nor closed access can completely satisfy the need of both user groups, we consider a shared access approach where the femtocell has a time-slot ratio $\eta$ between the home and
cellular users it serves, where $\eta=1$ is closed access. An optimal value of $\eta$ is found to maximize \textit{network throughput} subject to QoS requirements on the minimum throughput per home and cellular user. For example, given a cell edge femtocell with minimum throughput of 50 kbps/cellular user and 500 kbps/home user, this shared access approach achieves about 80\% higher network throughput than open access. When the QoS requirements increase in favor of significantly higher throughput of cellular user, shared access provides the lower network throughput than open access.

%By adjusting the value of $\eta$, the femtocell can provide appreciable throughput to both parties (e.g. on the order of 100 kbps/user for a cell edge femtocell). Under feasible throughput requirements of the home and cellular users, an optimal value of $\eta$ is found to maximize the network throughput. For example, given a cell edge femtocell with minimum throughput of 50 kbps/cellular user and 500 kbps/home user, this shared access approach achieves about 80\% higher network throughput than open access.
%We consider a shared access approach where the femtocell adjusts a time-slot ratio $\eta$ between the home and cellular users it serves, where $\eta = 1$ is closed access.  An optimal value of $\eta$ is found to maximize \textit{network throughput} subject to QoS requirements on the minimum throughput for both home and cellular users. For example, given a cell edge femtocell with minimum throughput of 50 kbps/cellular user and 500 kbps/home user, this shared access approach achieves about 80\% higher network throughput than open access while also satisfying the QoS requirements. When the QoS requirements are set too high, i.e. an optimal value of $\eta$ is infeasible, the femtocell limits the number of its serving cellular users to satisfy the QoS requirements.   %\textbf{(What happens when the QoS requirements are set too high and cannot be met?) }
\section{System Model}
\label{sec:model}
Denote $\mathcal{C}\subset \mathbb{R}^2$ as the circular interior of a macrocell with radius $R_\mathrm{c}$ and area $|\mathcal{C}|=\pi R_\mathrm{c}^2$ centered at a MBS. Since FAPs are installed by the end customer, they are distributed with randomness rather than regular pattern. FAPs $\{A_j\}_{j\in \Phi(\lambda)}$ are thus assumed to be distributed according to a homogeneous PPP with intensity $\lambda$, denoted $\Phi(\lambda)=\{X_j\}$, where each $X_j$ is the location of the $j$th FAP. The mean number of femtocells per cell cite is given as $N_\mathrm{f}=\lambda |\mathcal{C}|$.
Cellular users are assumed to be uniformly distributed inside $\mathcal{C}$. Home users are uniformly deployed in indoor (home) area, a disc of radius $R_\mathrm{i}$ centered at their FAP (see Fig. \ref{fig:coverage}). A summary of notation is given in Table \ref{table_1}.

\subsection{Channel Model and Multi-Level Modulation}
The downlink channel experiences path loss, Rayleigh fading with unit average power, and wall penetration loss $L<1$. The path loss exponents are denoted by $\alpha$ (outdoor and outdoor-to-indoor or indoor-to-outdoor) and $\beta$ (indoor-indoor). As in \cite{FemtoPower5},\cite{FemtoFreq3}, and \cite{ASE}, the downlink femtocell networks are assumed to be interference-limited and thermal noise at the receiver is ignored for simplicity. The MBS and FAP use fixed transmit powers of $P_\mathrm{c}$ and $P_\mathrm{f}$, respectively. We assume that orthogonal multiple access is used (TDMA or OFDMA on a per subband basis), thus no intracell interference is considered.
Interference from neighboring macrocell BSs is ignored for analytical tractability\footnote{Fig. \ref{fig:zoneTput} suggests that the assumption is not a significant omission of interference effects for dense femtocell systems. See Section \ref{sec:Result1} for more discussion}. We consider multi-level M-ary modulation single carrier transmission that is adapted to the received SIR $\gamma$, thus each user is assumed to estimate its SIR and provide perfect SIR feedback to their MBS (or FAP). Define $N$ SIR regions as $R_n=[\Gamma_n, \Gamma_{n+1}), n=1,\cdots,N$, where $\Gamma_1$ is the minimum SIR providing the lowest discrete rate and $\Gamma_{N+1}=\infty$. Then, the instantaneous transmission rate (in bps/Hz) is
\begin{equation}\label{eq:rate}
r_n=\log_2\left(1+\frac{\Gamma_n}{G}\right) \mathrm{~for~} \gamma\in R_n, 1\leq n\leq N,
\end{equation}
where $G$ is the Shannon Gap for multi-level M-ary modulation (and may assume some level of coding). Assuming round robin (RR) scheduling with equal time slots, the average throughput is
\begin{equation}\label{eq:Avgrate}
T = \sum_{n=1}^{N} r_n \cdot \mathbb{P}[\Gamma_n \leq \gamma < \Gamma_{n+1}].
\end{equation}

\subsection{Femtocell Coverage and Access}\label{sec:coverage}
We assume that all users are served by the station (MBS or FAP) from which they receive the strongest long-term average power as their serving stations. Therefore, a femtocell coverage area $\mathcal{F}\subset \mathbb{R}^2$ as the area with a border at which the long-term average received power from a central MBS and the FAP is the same. The coverage is given by the following lemma.
\begin{lemma}
For a FAP at distance $D$ from a central MBS located at the origin, the
border of femtocell coverage is a circle centered at $(\frac{\kappa^{2/\alpha}D}{\kappa^{2/\alpha}-1},0)$ and the radius $R_\mathrm{f}$ is
\begin{equation}\label{eq:Rf}
R_\mathrm{f}=\frac{\kappa^{1/\alpha}D}{|\kappa^{2/\alpha}-1|},
\end{equation}
where $\kappa=\frac{P_\mathrm{c}}{P_\mathrm{f} L}\neq 1$.
%\begin{equation}
%1-G(N-1)x^{N-1}\leq \mathbb{P}[\sin^2\theta_k\geq x]\leq 1-x^{N-1}
%, ~0\leq x\leq1.
%\end{equation}
\end{lemma}
\begin{proof} Consider a central MBS located at $(0,0)$ and an FAP at distance $D$ from the MBS. Without loss of generality, the FAP is then assumed to be located at $(D,0)$.
The received power at the position $(x,y)$ with distances $d_\mathrm{m}=\sqrt{x^2+y^2}$ and $d_\mathrm{c}=\sqrt{(D-x)^2+y^2}$ from MBS and FAP are respectively given as $P_\mathrm{r}^{(c)} = P_\mathrm{c}P_0(\frac{d_\mathrm{m}}{d_0})^{-\alpha}$ and $P_\mathrm{r}^{(f)} = P_\mathrm{f}L P_0(\frac{d_\mathrm{f}}{d_0})^{-\alpha}$, where $P_0$ is the path loss at a reference distance $d_0$. The contour with $P_\mathrm{r}^{(c)}=P_\mathrm{r}^{(f)}$ (zero dB SIR) satisfies $\frac{x^2+y^2}{(D-x)^2+y^2}=\kappa^{\alpha/2}$,
%\begin{equation}\label{eq:Lem1-1}
%\frac{x^2+y^2}{(D-x)^2+y^2}=\kappa^{\alpha/2},
%\end{equation}
where $\kappa = \frac{P_\mathrm{c}}{P_\mathrm{f} L}$. For $\kappa\neq1$ the equation is rewritten as $\left(x-\frac{\kappa^{2/\alpha}D}{\kappa^{2/\alpha}-1}\right)^2
+y^2=\frac{\kappa^{2/\alpha}D^2}{(\kappa^{2/\alpha}-1)^2}$,
%\begin{equation}\label{eq:Lem1-2}
%\left(x-\frac{\kappa^{2/\alpha}D}{\kappa^{2/\alpha}-1}\right)^2
%+y^2=\frac{\kappa^{2/\alpha}D^2}{(\kappa^{2/\alpha}-1)^2},
%\end{equation}
which is the equation of circle and completes the proof.
\end{proof}
The assumption $\kappa\neq1$ in Lemma 1 is valid for real scenario because where $P_\mathrm{c}>P_\mathrm{f}$ and $L<1$.
Since $\frac{\kappa^{2/\alpha}}{\kappa^{2/\alpha}-1}\approx1$ because $\kappa^{2/\alpha}\gg1$, we assume that the center of the cell coverage is equal to the FAP location. For example, $\frac{\kappa^{2/\alpha}}{\kappa^{2/\alpha}-1}=1.02$ for the values in Table \ref{table_1}. Lemma 1 states that the femtocell coverage $\mathcal{F}$ extends towards the cell edge. This further indicates that for an FAP close to the MBS, femtocell coverage can be smaller than indoor area, whereas for an FAP far from the MBS, the coverage leaks into outdoor area.

When a femtocell operates in closed access, only registered users (termed {\it home users}) can communicate with the femtocell, whereas in open access, unregistered users within the femtocell coverage (termed {\it neighboring cellular users}) as well as home users may connect to the femtocell. Assuming that neighboring cellular users are outdoors, we partition the macrocell into two regions, {\it inner region} and {\it outer region}, with the threshold distance $D_\mathrm{th}$ at which the femtocell coverage area is exactly equal to the indoor (home) area (see Fig. \ref{fig:coverage}). In the inner region, so some ``home users'' actually communicate with the MBS, while in the outer region, neighboring cellular users would like to connect to the FAP. Home and neighboring cellular users have different signal and interference models due to wall penetration loss. These mean that the SIR of all users needs to be modeled differently for the location of users (indoor or outdoor), the type of base station (MBS or FAP), and femtocell access strategy (open or closed). In order to analyze downlink performance for the femtocell access scheme, we thus define four geographic zones, whereby  users in the same zone have the same signal and interference model. Refer Fig. \ref{fig:coverage}.
\begin{itemize}
\item $\mathcal{F}_\mathrm{i}$: the indoor area (a disc with the radius $R_i$) covered by the FAP at $D> D_\mathrm{th}$.
\item $\mathcal{F}_\mathrm{o}$: the outdoor area (a circular annulus with inner radius $R_\mathrm{i}$ and outer radius $R_\mathrm{f}$) covered by the FAP at $D>D_\mathrm{th}$ in open access or covered by the MBS in closed access.
\item $\mathcal{F}_\mathrm{a}$: the indoor area (a disc with the radius $R_\mathrm{f}$) covered by the FAP at $D\leq D_\mathrm{th}$.
\item $\mathcal{F}_\mathrm{b}$: the indoor area (a circular annulus with inner radius $R_\mathrm{f}$ and outer radius $R_\mathrm{i}$ with respect to the FAP at $D\leq D_\mathrm{th}$) covered by the MBS.
\end{itemize}
The zone $\mathcal{F}_\mathrm{i}$ and $\mathcal{F}_\mathrm{o}$ have the property of $\mathcal{F}=\mathcal{F}_\mathrm{i} \cup \mathcal{F}_\mathrm{o}$ and $\mathcal{F}_\mathrm{i}\cap\mathcal{F}_\mathrm{o}=\emptyset$. The zone $\mathcal{F}_\mathrm{a}$ and $\mathcal{F}_\mathrm{b}$ have the property of  $\mathcal{F}_\mathrm{i}=\mathcal{F}_\mathrm{a}\cup \mathcal{F}_\mathrm{b}$ and $\mathcal{F}_\mathrm{a}\cap\mathcal{F}_\mathrm{b}=\emptyset$. Although the cell coverage model based on multiple geographic zones i.e. with multiple SIR model, is more intricate than conventional model \cite{FemtoPower5}, \cite{FemtoFreq3} with single SIR model for closed access only, it is essential for a comparative study of open, closed, and shared access.

\section{Per-Zone Average Throughput}\label{sec:SIR}
%As shown in \ref{sec:coverage}, according to the distance $D$ between FAP and MBS and femtocell access schemes, all the users divide into four groups located in the zone $\mathcal{F}_\mathrm{a}$, $\mathcal{F}_\mathrm{b}$, $\mathcal{F}_\mathrm{i}$, and $\mathcal{F}_\mathrm{o}$. Since each group user has different interference scenario and SIR, we analyze the throughput for the each zone. The per-tier throughput for femtocell access schemes is then derived based on the per-zone throughput in next section.
Consider a reference FAP $A_0$ at distance $D$ from a central MBS $B_0$, and its home users (or neighboring cellular users) at distance $R$ and $D_\mathrm{c}$ from $A_0$ and $B_0$, respectively. As shown in \ref{sec:coverage}, according to SIR model, all the users divide into four groups located in the zone $\mathcal{F}_\mathrm{a}$, $\mathcal{F}_\mathrm{b}$, $\mathcal{F}_\mathrm{i}$, and $\mathcal{F}_\mathrm{o}$. We analyze the throughput for the each zone, and then, based on it, derive per-tier throughput in next section. We assume small sized femtocell $R\ll D$ resulting $D_\mathrm{c}\approx D$.

\subsection{Cellular User In Zone $\mathcal{F}_\mathrm{o}$}\label{sec:cellSIR}
Since the neighboring cellular users want to hand off to the FAP, they communicate with $A_0$ (open access) or $B_0$ (closed access), which results in different SIR according to the access scheme as follows:
%\begin{equation}
%\gamma(R)=\left\{ \begin{array}{ll}
%\cfrac{P_\mathrm{c} g_0 D^{-\alpha}}{P_\mathrm{f} L h_0 R^{-\alpha}+\sum_{j \in \Phi\backslash{A_0}} P_\mathrm{f} L h_{j} |X_{j}|^{-\alpha}} & \textrm{~Closed Access}\\
%\cfrac{P_\mathrm{f} L h_0 R^{-\alpha}}{P_c g_0 D^{-\alpha}+\sum_{j \in \Phi\backslash{A_0}} P_\mathrm{f} L h_{j} |X_{j}|^{-\alpha}} & \textrm{~~~Open Access}
%\end{array} \right.,
%\label{eq:SIR_H}
%\end{equation}
\begin{numcases}
{X =}
\cfrac{P_\mathrm{c} g_0 D^{-\alpha}}{P_\mathrm{f} L h_0 R^{-\alpha}+\sum_{j \in \Phi\backslash{A_0}} P_\mathrm{f} L h_{j} |X_{j}|^{-\alpha}} &
\textrm{~Closed Access} \label{eq:SIR_FoCA}\\
\cfrac{P_\mathrm{f} L h_0 R^{-\alpha}}{P_c g_0 D^{-\alpha}+\sum_{j \in \Phi\backslash{A_0}} P_\mathrm{f} L h_{j} |X_{j}|^{-\alpha}} &
\textrm{~~~Open Access}\label{eq:SIR_FoOA}
\end{numcases}
where $g_{0}$ is the exponentially distributed channel power (with unit mean) from $B_0$. $h_{j}$ is the exponentially distributed channel power (with unit mean) from the interfering FAP $A_j$. $|X_j|$ denotes the distance between the user and $A_j$. The following lemma quantifies the user SIR for the zone $\mathcal{F}_\mathrm{o}$.
\begin{lemma}
The CDF of spatially averaged SIR over $\mathcal{F}_\mathrm{o}$, a circular annulus with outer radius $R_\mathrm{f}$ and inner radius $R_\mathrm{i}$, is given as follows:
%$P_{h,\mathrm{CA}}=\mathbb{E}_{R,\theta}[\mathrm{Pr}(\gamma_{h, \mathrm{CA}}\leq\Gamma|R\in[R_i, R_\mathrm{f}], \theta\in[0, 2\pi])]$, where
%\begin{enumerate}
  %\item Closed access:

1) Closed access:
\begin{align}\label{eq:S-CA}
S_{\mathrm{o}}^{\mathrm{CA}}(\Gamma)&=\mathbb{E}_R \left[\mathbb{P}[\gamma(R) \leq \Gamma]\right],~ R \in [R_\mathrm{i}, R_\mathrm{f}]\cr
&=\!
\!1\!-\!\frac{e^{-\lambda C_\alpha (K \Gamma)^{2/\alpha}}}{R_\mathrm{f}^2-R_\mathrm{i}^2}\!\left(\!R_\mathrm{f}^2\!-\!R_\mathrm{i}^2\!+\!R_\mathrm{i}^2
~\!_2F_1\!\left[\tfrac{2}{\alpha},1;1\!+\!\tfrac{2}{\alpha};
\tfrac{-R_\mathrm{i}^\alpha}{K \Gamma}\right]%\cr
%&~&
\!-\!R_\mathrm{f}^2~\!_2F_1\!\left[\tfrac{2}{\alpha},1;1\!+\!\tfrac{2}{\alpha};
\tfrac{-R_\mathrm{f}^\alpha}{K \Gamma}\right]\right)
\end{align}
where $C_\alpha=\frac{2\pi^2}{\alpha}\csc(\frac{2\pi}{\alpha})$, $K=\frac{P_\mathrm{f}LD^{\alpha}}{P_\mathrm{c}}$, and $_2F_1[\cdot]$ is the Gauss hypergeometric function.

  %\item Open access:
2) Open access:
%\begin{equation}\label{eq:S-OA}
%S_{\mathcal{F}_o}^{\mathrm{OA}}=
%1-\frac{1}{R_\mathrm{f}^2-R_i^2}\int_{R_i^2}^{R_\mathrm{f}^2}
%\frac{e^{-r\lambda_fC\Gamma^{2/\alpha}}}{\kappa\Gamma D^{-\alpha}r^{\alpha/2}+1}\mathrm{d}r.
%\end{equation}
\begin{equation}\label{eq:S-OA}
S_{\mathrm{o}}^{\mathrm{OA}}(\Gamma)=
1-\frac{1}{R_\mathrm{f}^2-R_\mathrm{i}^2}\int_{R_\mathrm{i}^2}^{R_\mathrm{f}^2}
\frac{e^{-r\lambda C_\alpha\Gamma^{2/\alpha}}}{K^{-1} \Gamma  r^{\alpha/2}+1}\mathrm{d}r.
\end{equation}
If $\alpha=4$, then
\begin{equation}\label{eq:S-OA1}
S_{\mathrm{o}}^{\mathrm{OA}}(\Gamma)= 1-\frac{B(R_\mathrm{f}^2)-B(R_\mathrm{i}^2)}{(R_\mathrm{f}^2-R_\mathrm{i}^2)},
\end{equation}
\begin{equation*}\label{eq:B(x)}
B(x)\!=\!\tfrac{1}{\sqrt{z}}\!\Big[\!\big(\!-\!\mathrm{Re}\{\mathrm{Ei}(iw)\}
\!+\!\mathrm{Re}\{\mathrm{Ei}(xy\!+\!iw)\}\big)\!\sin(w)
\!+\!\big(\mathrm{Im}\{\mathrm{Ei}(iw)\}\!-\!\mathrm{Im}\{\mathrm{Ei}(xy\!+\!iw)\}\!\big)\!\cos(w)\!\Big]\!,
\end{equation*}
%\begin{eqnarray}\label{eq:B(x)}
%with~~B(x)&=&\frac{1}{\sqrt{z}}\Big[\big(-\mathrm{Re}\{\mathrm{Ei}(iw)\}
%+\mathrm{Re}\{\mathrm{Ei}(xy+iw)\}\big)\sin(w)\cr
%&~&+\big(\mathrm{Im}\{\mathrm{Ei}(iw)\}-\mathrm{Im}\{\mathrm{Ei}(xy+iw)\}\big)\cos(w)\Big],
%\end{eqnarray}
\begin{equation}
y = -\lambda C_\alpha\sqrt{\Gamma}, ~ z=\Gamma K^{-1}, ~ w=y/\sqrt{z}=-\lambda C_\alpha\sqrt{K},
\end{equation}
where $\mathrm{Re}\{z\}$ and $\mathrm{Im}\{z\}$ represent the real and imaginary parts of $z$, respectively, $\mathrm{Ei}(z)=\int_{-\infty}^{z}\frac{e^{-t}}{t}dt$ is the Exponential integral function.
%\end{enumerate}
%\begin{equation}
%1-G(N-1)x^{N-1}\leq \mathbb{P}[\sin^2\theta_k\geq x]\leq 1-x^{N-1}
%, ~0\leq x\leq1.
%\end{equation}
\end{lemma}
\begin{proof} See Appendix~\ref{sec:PfLem2}.
\end{proof}

\subsection{Home User SIR In Zone $\mathcal{F}_\mathrm{i}$ or $\mathcal{F}_\mathrm{a}$}\label{sec:homeSIR}
The home user SIR for the zone $\mathcal{F}_\mathrm{i}$ is given as
\begin{equation}
\gamma(R) = \frac{P_\mathrm{f} h_0 R^{-\beta}}{P_\mathrm{c} L g_0 D^{-\alpha}+\sum_{j \in \Phi\backslash{A_0}} P_\mathrm{f} L^2 h_{j} |X_{j}|^{-\alpha}}
,\label{eq:SIR_Fi}
\end{equation}
This is the same as the cellular user SIR for open access given in (\ref{eq:SIR_FoOA}) except for the distinction that this home user is indoors and so the propagation terms are adjusted accordingly. The SIR distribution of the home user in $\mathcal{F}_\mathrm{i}$ is given in the following lemma.
\begin{lemma}
The CDF of spatially averaged SIR over the zone $\mathcal{F}_\mathrm{i}$, a disc with radius $R_\mathrm{i}$, is given as
\begin{align}\label{eq:S-F}
S_{\mathrm{i}}(\Gamma)
&=\mathbb{E}_R \left[\mathbb{P}[\gamma(R) \leq \Gamma]\right],~ R \in [0, R_\mathrm{i}]\cr
&=1-\frac{2}{R_\mathrm{i}^2}\int_{0}^{R_\mathrm{i}}
\frac{r\cdot e^{-\lambda C_\alpha (L^2\Gamma)^{2/\alpha}r^{2\beta/\alpha}}}{K\Gamma r^{\beta}+1}dr,
\end{align}
where $K=\frac{P_\mathrm{c} L}{P_\mathrm{f} D^{\alpha}}$.
If $\alpha=4$ and $\beta=2$, then
\begin{equation}\label{eq:S-F1}
S_{\mathrm{i}}(\Gamma)= 1-H(R_\mathrm{i})/R_\mathrm{i}^2,
\end{equation}
\begin{equation*}
H(x)\!=\!\tfrac{2}{z}\left[\left(\mathrm{Re}\{\mathrm{Ei}(xy\!+\!iw)\}
\!-\!\mathrm{Re}\{\mathrm{Ei}(iw)\}\right)\cos(w)%\cr
\!+\!\left(\mathrm{Im}\{\mathrm{Ei}(xy+iw)\}\!-\!\mathrm{Im}\{\mathrm{Ei}(iw)\}
\right)\sin(w)\right],
\end{equation*}
\begin{equation}\label{eq:S-F1-1}
y= -\lambda C_\alpha\sqrt{\Gamma}L, ~ z=K\Gamma , ~ w=y/\sqrt{z}=-\lambda C_\alpha L/\sqrt{K},
\end{equation}
Or, if $\alpha=\beta=4$, then
\begin{equation}\label{eq:S-F2}
S_{\mathrm{i}}(\Gamma)= 1-B(R_\mathrm{i}^2)/R_\mathrm{i}^2,
\end{equation}
where $w$, $y$, and $z$ for calculating $B(R_\mathrm{i}^2)$ are given in (\ref{eq:S-F1-1}).
\end{lemma}
\begin{proof} See Appendix~\ref{sec:PfLem4}.
\end{proof}
Note that the user SIR in the zone $\mathcal{F}_\mathrm{a}$ is given in (\ref{eq:SIR_Fi}), since the zone is the indoor area covered by a FAP at $D\leq D_\mathrm{th}$ (in the inner region). The SIR distribution of the home users in $\mathcal{F}_\mathrm{a}$ is given in the following corollary.
\begin{coll}
The user SIR in the zone $\mathcal{F}_\mathrm{a}$ is given by (\ref{eq:SIR_Fi}) and $\mathcal{F}_\mathrm{a}$ is a disc with radius $R_\mathrm{f}$. Thus, the CDF of spatially averaged SIR over the zone $\mathcal{F}_\mathrm{a}$ is given by (\ref{eq:S-F}) with $R_f$ replacing $R_i$, i.e.,
\begin{equation}\label{eq:S-Fa}
S_{\mathrm{a}}(\Gamma) = S_{\mathrm{i}}(\Gamma)|_{R_i=R_f}.
\end{equation}
\end{coll}

\subsection{Home User SIR In Zone $\mathcal{F}_\mathrm{b}$}
The zone $\mathcal{F}_\mathrm{b}$ exists for the femtocells with $D\leq D_\mathrm{th}$ (in the inner region). The user SIR is
\begin{equation}
\gamma(R) = \frac{P_\mathrm{c} g_0 L D^{-\alpha}}{P_\mathrm{f} h_0 R^{-\beta}+\sum_{j \in \Phi\backslash{A_0}} P_\mathrm{f} L^2 h_{j} |X_{j}|^{-\alpha}}
.\label{eq:SIR_Fb}
\end{equation}
The SIR is the same as the cellular user SIR for closed access given in (\ref{eq:SIR_FoCA}) except for the distinction that this home user is indoors and so the propagation terms are adjusted accordingly. Thus, the SIR distribution for $\mathcal{F}_\mathrm{b}$ is given in similar form to (\ref{eq:S-CA}) as shown in the following lemma.
\begin{lemma}
The CDF of spatially averaged SIR over the area $\mathcal{F}_\mathrm{b}$, a circular annulus with outer radius $R_\mathrm{i}$ and inner radius $R_\mathrm{f}$, is given as
\begin{align}\label{eq:S-Fb}
S_{\mathrm{b}}(\Gamma)&=\mathbb{E}_R \left[\mathbb{P}[\gamma(R) \leq \Gamma]\right],~ R \in [R_\mathrm{f}, R_\mathrm{i}] \cr
&=
1\!-\!\frac{e^{-\lambda C_\alpha (L^2 K \Gamma )^{2/\alpha}}}{R_\mathrm{i}^2-R_\mathrm{f}^2}  \!\left( \! R_\mathrm{i}^2\!-\!R_\mathrm{f}^2 \! + \! R_\mathrm{f}^2
~_2F_1\!\left[\!\tfrac{2}{\beta},1;1\!+\!\tfrac{2}{\beta};
 \!\tfrac{ -R_\mathrm{f}^\beta}{K \Gamma}\!\right]\!%\cr
%&~&
\!-\!R_\mathrm{i}^2~_2F_1\!\left[\!\tfrac{2}{\beta},1;1 \! + \! \tfrac{2}{\beta};
\!\tfrac{ -R_\mathrm{i}^\beta}{K \Gamma}\!\right]\!\right)\!.
\end{align}
\end{lemma}
where $K=\frac{P_\mathrm{f} D^{\alpha}}{P_\mathrm{c} L}$.
\begin{proof} See Appendix~\ref{sec:PfLem5}.
\end{proof}
Finally, combining equation (\ref{eq:Avgrate}) and the SIR distribution ((\ref{eq:S-CA}), (\ref{eq:S-OA}), (\ref{eq:S-F}), (\ref{eq:S-Fa}), and (\ref{eq:S-Fb})), the average throughput of each zone $\mathcal{F}_\mathrm{a}$, $\mathcal{F}_\mathrm{b}$, $\mathcal{F}_\mathrm{i}$, and $\mathcal{F}_\mathrm{o}$ is given as
\begin{equation}\label{eq:T-Fall}
T_{x} = \sum_{n=1}^{N} r_n \left[S_{x} (\Gamma_{n+1}) - S_{x} (\Gamma_{n})\right],~~x\in\{\mathrm{a},\mathrm{b},\mathrm{i},\mathrm{o}\}
\end{equation}

\subsection{Numerical Results}\label{sec:Result1}
Fig. \ref{fig:zoneTput} shows the spatially averaged throughput for zone $\mathcal{F}_\mathrm{a}$, $\mathcal{F}_\mathrm{b}$, $\mathcal{F}_\mathrm{i}$, and $\mathcal{F}_\mathrm{o}$ versus FAP $A_0$ of distance $D$ from the MBS $B_0$ for a different number of femtocells. Here the system parameters in Table \ref{table_1} are used. For the zone $\mathcal{F}_\mathrm{o}$, we show two results $T_\mathrm{o}^\mathrm{CA}$ (closed access) and $T_\mathrm{o}^\mathrm{OA}$ (open access). The analytic curves given from (\ref{eq:T-Fall}) are very close to the simulated curves. Furthermore, there is not a considerable difference in throughput between the inter-macrocell interference (marked with ``$\circ$'' and ``$\Box$'') and no inter-macrocell interference (marked with ``$\times$'' and ``$+$''), which validates our assumption that neglecting inter-macrocell interference is acceptable in dense deploy femto.

The higher interference from neighboring femtocells caused by higher $N_\mathrm{f}$ decreases the averaged throughput of all the zones. On the other hand, $D$ has different effects in different zones as follows. First, since increasing $D$ decreases the signal power from $B_0$, the average throughput ($T_\mathrm{b}$ and $T_\mathrm{o}^\mathrm{CA}$) served by the $B_0$ decreases with $D$.
Second, the area of zones $\mathcal{F}_\mathrm{a}$ and $\mathcal{F}_\mathrm{o}$ increases with $D$, which enlarges the area with low SIR within the zone, i.e. it has a negative effect on SIR. However, for the users in $\mathcal{F}_\mathrm{a}$ and $\mathcal{F}_\mathrm{o}$ (open access), the interference power from $B_0$ decreases with $D$, which enhances SIR. As a result, $T_\mathrm{a}$ decreases with increasing $D$ indicating that increasing the area of $\mathcal{F}_\mathrm{a}$ counteracts the effects of decreasing cross-tier interference from $B_0$. Moreover, $T_\mathrm{o}^\mathrm{OA}$ increases for small $D$, but decreases for large $D$, which indicates that a negative effect of increasing the area of $\mathcal{F}_\mathrm{o}$ on $T_\mathrm{o}^\mathrm{OA}$ becomes dominant at higher $D$.
Third, for the zone $\mathcal{F}_\mathrm{i}$, no negative effect on SIR is observed since the zone area is fixed (independent of $D$). Thus, $T_\mathrm{i}$ monotonically increases with $D$.

\section{Per-Tier User Throughput: Closed Access vs. Open Access}\label{sec:Tput}
In this section, we analyze the per-tier user throughput of closed and open access based on the number of users in the zone as well as the per-zone throughput obtained in preceding section. Denote $U_\mathrm{a}$, $U_\mathrm{b}$, $U_\mathrm{i}$, and $U_\mathrm{o}$ is the number of users in the zone $\mathcal{F}_\mathrm{a}$, $\mathcal{F}_\mathrm{b}$, $\mathcal{F}_\mathrm{i}$, and $\mathcal{F}_\mathrm{o}$, respectively. Let $U_\mathrm{c}$ and $U_\mathrm{h}$ denote the number of outdoor cellular users and the number of home users per femtocell, respectively. When $U_\mathrm{h}$ is assumed to be fixed for $N_\mathrm{f}$ femtocells, total average number of user in a cell cite is then given by
\begin{eqnarray}\label{eq:U-tot}
U= U_\mathrm{c} + N_\mathrm{f}U_\mathrm{h}%\cr
= U_\mathrm{c} + (N_\mathrm{f1}+N_\mathrm{f2})U_\mathrm{h}%\cr
\mathop =\limits^{\left( a \right)} U_\mathrm{c} + N_\mathrm{f1}(U_\mathrm{a}+U_\mathrm{b})+N_\mathrm{f2}U_\mathrm{i},
\end{eqnarray}
where $N_\mathrm{f1}=N_\mathrm{f} (D_\mathrm{th}/R_\mathrm{c})^2$ and $N_\mathrm{f2}=N_\mathrm{f} (1-(D_\mathrm{th}/R_\mathrm{c})^2)$ is respectively the average number of femtocells in the inner region ($D\leq D_\mathrm{th}$) and the outer region ($D > D_\mathrm{th}$). Furthermore, (a) follows from that $U_\mathrm{h}=U_\mathrm{i}=U_\mathrm{a}+U_\mathrm{b}$.

\subsection{Closed Access}\label{sec:TputCA}
Consider a reference FAP in closed access at distance $D$ from a central MBS. For the inner region, home users in $\mathcal{F}_\mathrm{i}$ connect to the FAP, while the neighboring cellular users of the FAP (implying users in $\mathcal{F}_\mathrm{o}$) are served by the MBS. For the outer region, regardless of femtocell access scheme, the home users in $\mathcal{F}_\mathrm{a}$ connect to the FAP, but the remaining home users in $\mathcal{F}_\mathrm{b}$ communicate to the MBS. The home users in $\mathcal{F}_\mathrm{b}$ share the same frequency channel with cellular users by using different time slots. Based on the femtocell/macrocell access scenario of the users, the following theorem quantifies per-tier user throughput in closed access.

\begin{thm}
In closed access, the average sum throughput of home users $T_\mathrm{h}^{\mathrm{CA}}$ and neighboring cellular users $T_\mathrm{c}^{\mathrm{CA}}$ with respect to a FAP at distance $D$ from a central MBS is given as
\begin{eqnarray}
T_\mathrm{h}^{\mathrm{CA}}(D) &=& \left\{ \begin{array}{ll}
T_\mathrm{a}(D) + \rho_\mathrm{b}^{\mathrm{CA}} T_\mathrm{b}(D) & D \leq D_\mathrm{th}\\
T_\mathrm{i}(D) & D > D_\mathrm{th}
\end{array} \right.\label{eq:T-f-CA}\\
T_\mathrm{c}^{\mathrm{CA}}(D) &=& \rho_\mathrm{o} T_\mathrm{o}^{\mathrm{CA}}(D),~~ D > D_\mathrm{th}
\label{eq:T-c-CA}
\end{eqnarray}
where the per-zone throughput $T_\mathrm{a}(D)$, $T_\mathrm{b}(D)$, $T_\mathrm{i}(D)$, and $T_\mathrm{o}^{\mathrm{CA}}(D)$ is given from (\ref{eq:T-Fall}). $\rho_\mathrm{b}^{\mathrm{CA}}$ and $\rho_\mathrm{o}$ is the {\it fraction of time-slot} dedicated to the home users in $\mathcal{F}_\mathrm{b}$ and the cellular users in $\mathcal{F}_\mathrm{o}$, respectively, among all users supported by the MBS, which is given as
\begin{align}
\rho_\mathrm{b}^{\mathrm{CA}} &= \frac{U_\mathrm{h} \left(1-2\mathcal{K}D^2\right)} {U_\mathrm{c} + N_\mathrm{f1} U_\mathrm{h} \left(1-\mathcal{K}D_\mathrm{th}^2\right)},
\label{eq:thm1rho-CA}\\
\rho_\mathrm{o} &= \frac{U_\mathrm{c} R_\mathrm{i}^2 (2\mathcal{K}D^2-1)}{(R_\mathrm{c}^2-N_\mathrm{f}R_\mathrm{i}^2)
(U_\mathrm{c} + N_\mathrm{f1} U_\mathrm{h} (1-\mathcal{K}D_\mathrm{th}^2)) },
\label{eq:thm1rho-o}
\end{align}
where $\mathcal{K}=\tfrac{\kappa^{2/\alpha} }{ 2 R_\mathrm{i}^2 (\kappa^{2/\alpha}-1)^2 }$.
\end{thm}
\begin{proof} See Appendix~\ref{sec:PfThm1}.
\end{proof}

\begin{rmk}
From (\ref{eq:thm1rho-CA}) and Fig. \ref{fig:zoneTput}, increasing $D$ enhances $T_\mathrm{i}$ but reduces $\rho_\mathrm{b}^{\mathrm{CA}}$, $T_\mathrm{a}$, and $T_\mathrm{b}$. Therefore, in (\ref{eq:T-f-CA}) the home user throughput in closed access $T_\mathrm{h}^{\mathrm{CA}}$ decreases with $D$ in the inner region ($D\leq D_\mathrm{th}$) but increases with $D$ in the outer region ($D>D_\mathrm{th}$). Intuitively, the signal from the MBS is interference to home users in the outer region, but it is  the desired signal to some home users connecting to the MBS in the inner region. This results in throughput degradation (inner region) or improvement (outer region) by increasing $D$.
\end{rmk}
\begin{rmk}
From (\ref{eq:thm1rho-CA}), increasing $U_\mathrm{c}$ reduces $\rho_\mathrm{b}^{\mathrm{CA}}$, and $U_\mathrm{c}$ does not effect on $T_\mathrm{i}$, $T_\mathrm{a}$, and $T_\mathrm{b}$. Intuitively, many cellular users increase the MBS load and thereby the amount of radio resources allocated home user served by the MBS is decreased. This indicates that for femtocells in the inner region, $T_\mathrm{h}^{\mathrm{CA}}$ in (\ref{eq:T-c-CA}) is higher for a lower cellular user density, while it is independent of cellular user density in the outer region.
Since $\rho_\mathrm{o}$ increases with $U_\mathrm{c}$ in (\ref{eq:thm1rho-o}) and $U_\mathrm{c}$ does not effect on $T_\mathrm{o}^\mathrm{CA}$,  $T_\mathrm{c}^{\mathrm{CA}}$ in (\ref{eq:T-c-CA}) is high at a high cellular user density.
\end{rmk}
\subsection{Open Access}\label{sec:TputOA}
In the outer region, the reference FAP in open access provides service to neighboring cellular users in $\mathcal{F}_\mathrm{o}$ as well as home users in $\mathcal{F}_\mathrm{i}$. Thus, the home users share the downlink radio resource of the FAP with the cellular users in time division manner. On the other hand, in the inner region, the femtocell/macrocell access scenario of the home users in open access is the same as that in closed access. The following theorem quantifies per-tier user throughput in open access.
\begin{thm}
In open access, the average sum throughput of home users, $T_\mathrm{h}^{\mathrm{OA}}$, and neighboring cellular users, $T_\mathrm{c}^{\mathrm{OA}}$, with respect to a FAP at distance $D$ from a central MBS is given as
\begin{eqnarray}
T_\mathrm{h}^{\mathrm{OA}}(D)&=& \left\{ \begin{array}{ll}
T_\mathrm{a}(D) + \rho_\mathrm{b}^{\mathrm{OA}} T_\mathrm{b}(D) & D \leq D_\mathrm{th}\\
\rho_\mathrm{i} T_\mathrm{i}(D) & D > D_\mathrm{th}
\end{array} \right.\label{eq:T-f-OA}\\
T_\mathrm{c}^{\mathrm{OA}}(D) &=& (1-\rho_\mathrm{i}) T_\mathrm{o}^{\mathrm{OA}}(D),~~ D > D_\mathrm{th}
\label{eq:T-c-OA}
\end{eqnarray}
where $T_\mathrm{o}^{\mathrm{OA}}(D)$ is given from (\ref{eq:T-Fall}). $\rho_\mathrm{b}^{\mathrm{OA}}$ is the {\it fraction of time-slot} dedicated to the home users in $\mathcal{F}_\mathrm{b}$ among home and cellular users supported from the MBS, and $\rho_\mathrm{i}$ is the {\it fraction of time-slot} dedicated to the home users in $\mathcal{F}_\mathrm{i}$ among home and cellular users supported from the FAP. They are given as
\begin{align}
\rho_\mathrm{b}^{\mathrm{OA}} &= \frac{U_\mathrm{h} \left(1-2\mathcal{K}D^2\right)} {U_\mathrm{c} + N_\mathrm{f1} U_\mathrm{h} \left(1-\mathcal{K}D_\mathrm{th}^2\right)-N_\mathrm{f2} (\mathcal{K}(D_\mathrm{th}^2+R_\mathrm{c}^2)-1)
\frac{U_\mathrm{c}R_\mathrm{i}^2}{R_\mathrm{c}^2-N_\mathrm{f}R_\mathrm{i}^2}
}
\label{eq:thm2rho-OA}\\
\rho_\mathrm{i} &= \frac{U_\mathrm{i}} {U_\mathrm{i} + U_\mathrm{o} }= \left( 1+ \frac{U_\mathrm{c}R_\mathrm{i}^2(\mathcal{K}D^2-1)}
{U_\mathrm{i}(R_\mathrm{c}^2-N_\mathrm{f}R_\mathrm{i}^2)} \right)^{-1}
\label{eq:thm2rho-i},
\end{align}
\end{thm}
\begin{proof} See Appendix~\ref{sec:PfThm2}.
\end{proof}

\begin{rmk}
First, in (\ref{eq:T-f-OA}), $T_\mathrm{h}^{\mathrm{OA}}$ decreases with $D$ ($D\leq D_\mathrm{th}$), since increasing $D$ lowers $\rho_\mathrm{b}^{\mathrm{OA}}$ in (\ref{eq:thm2rho-OA}), and $T_\mathrm{a}$ and $T_\mathrm{b}$ decrease with $D$ from Fig. \ref{fig:zoneTput}. From (\ref{eq:thm2rho-i}), increasing $D$ reduces $\rho_\mathrm{i}$. In Fig. \ref{fig:zoneTput} an increment in $T_\mathrm{i}$ decreases with $D$. Intuitively, $T_\mathrm{i}$ is upper limited by the highest rate of M-ary modulation at large $D$, while $\rho_\mathrm{i}$ is lower limited by zero. This indicates that $T_\mathrm{h}^{\mathrm{OA}}$ begins to decrease at sufficiently large $D$. Second, from (\ref{eq:thm2rho-i}), $1-\rho_\mathrm{i}$ increases with $U_\mathrm{c}$. Thus, $T_\mathrm{c}^{\mathrm{OA}}$ in (\ref{eq:T-c-OA}) is enhanced at a higher cellular user density.
\end{rmk}
The throughput comparison of both the access schemes is given in the remarks below.
\begin{rmk}
First, closed access rather than open access increases the number of users supported by the MBS and thereby $\rho_\mathrm{b}^{\mathrm{CA}} \leq \rho_\mathrm{b}^{\mathrm{OA}}$. This is followed by $T_\mathrm{h}^{\mathrm{CA}} \leq T_\mathrm{h}^{\mathrm{OA}}$ for the inner region. Intuitively, open access femtocells in the outer region admit neighboring cellular users which reduces the macrocell load. This effectively increases the throughput of home users in inner region, which are supported by the MBS. On the other hand, $T_\mathrm{h}^{\mathrm{CA}} \geq T_\mathrm{h}^{\mathrm{OA}}$ for the outer region because $\rho_\mathrm{i}\leq 1$ from (\ref{eq:thm2rho-i}). Note also from (\ref{eq:thm2rho-i}) that $1-\rho_\mathrm{i}=\frac{U_\mathrm{o}} {U_\mathrm{i} + U_\mathrm{o} }$. Since $U_\mathrm{i} + U_\mathrm{o} < U_\mathrm{c} + N_\mathrm{f1} \overline{U}_\mathrm{b}$, comparing $1-\rho_\mathrm{i}$ and (\ref{eq:rho-o}) yields  $1-\rho_\mathrm{i} > \rho_\mathrm{o}$. Moreover, $T_\mathrm{o}^\mathrm{OA}$ is obviously larger than $T_\mathrm{o}^\mathrm{CA}$, thus $T_\mathrm{c}^\mathrm{CA} < T_\mathrm{c}^\mathrm{OA}$. This indicates that home and cellular users prefer opposite access schemes.
\end{rmk}
\subsection{Numerical Results}\label{sec:TputOA}
The throughput results in this section are obtained with the system parameters in Table \ref{table_1}. Fig. \ref{fig:TputF} shows the home user throughput analytically obtained using (\ref{eq:T-f-CA}) and (\ref{eq:T-f-OA}) versus FAP-MBS distance $D$ for different numbers of femtocells $N_\mathrm{f}$ and cellular users $U_\mathrm{c}$. Since $R_\mathrm{f}=R_\mathrm{i}$ at $D=D_\mathrm{th}$, we obtain the distance $D_\mathrm{th}$ of 130m by substituting $R_\mathrm{i}=20$ into (\ref{eq:Rf}).
In closed access, home user throughput decreases with $D$ $(D\leq D_\mathrm{th})$, while it increases with $D$ $(D>D_\mathrm{th})$ per Remark 1. For $D>D_\mathrm{th}$ the home user throughput in open access first increases then decreases with increasing $D$. Additionally, Fig. \ref{fig:TputFwrtUc} shows that the turning point moves into the cell interior with increasing $U_\mathrm{c}$. This is because increasing $D$ and $U_\mathrm{c}$ increases the number of neighboring cellular users, and thus, the time resource allocated to home user in femtocell downlink is reduced.
In Fig. \ref{fig:TputFwrtNf}, the throughput for both open and closed access is degraded, since the aggregated interference from other femtocells increases with $N_\mathrm{f}$.
We observe that unlike the case of $D>D_\mathrm{th}$, open access outperforms closed access for $D\leq D_\mathrm{th}$. However, the throughput loss of home users at $D\leq D_\mathrm{th}$ dominates the home user throughput. Thus, closed access is better for home users.

Fig. \ref{fig:TputC} plots the sum throughput of neighboring cellular users of a reference femtocell using (\ref{eq:T-c-CA}) and (\ref{eq:T-c-OA}). The throughput is high at a low femtocell density and a high cellular user density, which agrees with the prediction in Remark 2 and 3. The throughput ($T_\mathrm{c}^\mathrm{CA} < 0.003$ bps/Hz) in closed access is too low to offer typical services (0.003 bps/Hz is equivalent to 15 kbps for 5 MHz bandwidth). Thus, open access is much better for neighboring cellular users, in contrast to the result for home users. Table \ref{table_2} summarizes these results of closed vs. open access.

Fig. \ref{fig:TputN} plots the network throughput in open access and closed access, sum of home user and neighboring cellular user throughput, i.e., $T^\mathrm{CA}=T_\mathrm{h}^\mathrm{CA}+T_\mathrm{c}^\mathrm{CA}
=T_\mathrm{i}+\rho_\mathrm{o}T_\mathrm{o}^\mathrm{CA}$ and $T^\mathrm{OA}=T_\mathrm{h}^\mathrm{OA}+T_\mathrm{c}^\mathrm{OA}
=\rho_\mathrm{i}T_\mathrm{i}+(1-\rho_\mathrm{i})T_\mathrm{o}^\mathrm{OA}$. Note that with respect to the network throughput for $D>D_\mathrm{th}$, open access is inferior to closed access. The reason is from the inequality given by $T^\mathrm{CA}-T^\mathrm{OA}=(1-\rho_\mathrm{i})(T_\mathrm{i}-T_\mathrm{o}^\mathrm{OA})
+\rho_\mathrm{o}T_\mathrm{o}^\mathrm{CA}>0$. Intuitively, since  $T_\mathrm{i}-T_\mathrm{o}^\mathrm{OA}>T_\mathrm{o}^\mathrm{OA}-T_\mathrm{o}^\mathrm{CA}$ from Fig. \ref{fig:zoneTput}, the decrement of home user throughput $T_\mathrm{i}$ due to time resource sharing with cellular users in open access prevails against the increment of cellular user throughput by substituting $T_\mathrm{o}^\mathrm{OA}$ for $T_\mathrm{o}^\mathrm{CA}$. In a different point of view, this implies that a slight increase in the time fraction $\rho_\mathrm{i}$ provides a high increase $T_\mathrm{h}^\mathrm{OA}$ at the cost of a slight drop in $T_\mathrm{c}^\mathrm{OA}$, i.e., an increase in network throughput $T^\mathrm{OA}$. This, as well as the extremely low throughput in closed access motivates the shard access femtocellls using time slot allocation, which will be in the next section.

\section{Shared Access: Time-slot Allocation}\label{sec:HA}
We consider the hybrid access where a FAP allocates $\eta$ fraction of time-slots to home users and the remaining $1-\eta$ fraction of time-slots to cellular users. Unlike the time-slot allocation in open access, where the time fraction $\rho_\mathrm{i}$ is dependent on the number of home users and cellular users, the time-slot allocation in the shared access optimizes $\eta$ to maximize the network throughput $T^\mathrm{SA}$ while satisfying QoS requirement.
The network throughput $T^\mathrm{SA}$ is defined as
\begin{equation}
T^\mathrm{SA} = \eta T_\mathrm{i} + (1-\eta) T_\mathrm{o}^\mathrm{OA}
, ~~ \eta \in [0,1].
\end{equation}
We define the QoS requirement as follows: 1) The average user throughput $\overline{T}_\mathrm{c}$ (cellular user) and $\overline{T}_\mathrm{h}$ (home user) is larger than the required minimum throughput $\Omega_\mathrm{c}$ (cellular user) and $\Omega_\mathrm{h}$ (home user), respectively, and 2) The average user throughput $\overline{T}_\mathrm{c}$ is at least $\varepsilon\in [0,1]$ w.r.t  the $\overline{T}_\mathrm{h}$.
Satisfying the QoS, the time-slot allocation problem to maximize the network throughput $T^\mathrm{SA}$ is formulated as
\begin{eqnarray}
\max_{0\leq\eta \leq 1} &\eta T_\mathrm{i} + (1-\eta) T_\mathrm{o}^\mathrm{OA} \label{eq:Timeslot}\\
\mathrm{subject~to} &\overline{T}_\mathrm{c} \geq \Omega_\mathrm{c},~
\overline{T}_\mathrm{h} \geq \Omega_\mathrm{h},
~ \overline{T}_\mathrm{c} \geq \varepsilon \overline{T}_\mathrm{h}\label{eq:Qos3}
\end{eqnarray}
where $\overline{T}_\mathrm{h}=\frac{\eta T_\mathrm{i}}{U_\mathrm{i}}$ and $\overline{T}_\mathrm{c}=\frac{(1-\eta) T_\mathrm{o}^\mathrm{OA}}{U_\mathrm{o}}$.
\begin{prop}
The optimal value $\eta^\ast$ of the time-slot allocation in (\ref{eq:Timeslot}) is given as
\begin{equation}\label{eq:Prop1}
\eta^\ast = \min \left( 1-\frac{\Omega_\mathrm{c}U_\mathrm{o}}{T_\mathrm{o}^\mathrm{OA}},~
\left(1+\varepsilon\frac{U_\mathrm{o}T_\mathrm{i}}{U_\mathrm{i}T_\mathrm{o}^\mathrm{OA}}\right)^{-1}
\right ),
\end{equation}
The solution $\eta^\ast$ is feasible when it is equal or larger than $\frac{\Omega_\mathrm{h}U_\mathrm{i}}{T_\mathrm{i}}$.
\end{prop}
\begin{proof}
Denote $\mathcal{Q}_1$, $\mathcal{Q}_2$, and $\mathcal{Q}_3$ as a set of $\eta$ satisfying the QoS requirement (\ref{eq:Qos3}) in the order of description, respectively. we then obtain intersection of the three sets as $\mathcal{Q} = \mathcal{Q}_1\cap \mathcal{Q}_2 \cap \mathcal{Q}_3 =
\Big\{ \eta \mid \frac{\Omega_\mathrm{h}U_\mathrm{i}}{T_\mathrm{i}} \leq \eta \leq \min \left( 1-\frac{\Omega_\mathrm{c}U_\mathrm{o}}{T_\mathrm{o}^\mathrm{OA}},~
\left(1+\varepsilon\frac{U_\mathrm{o}T_\mathrm{i}}{U_\mathrm{i}T_\mathrm{o}^\mathrm{OA}}\right)^{-1}
 \right) \Big \}$ for $\frac{\Omega_\mathrm{h}U_\mathrm{i}}{T_\mathrm{i}}\leq
\min  \left( 1-\frac{\Omega_\mathrm{c}U_\mathrm{o}}{T_\mathrm{o}^\mathrm{OA}},~
\left(1+\varepsilon\frac{U_\mathrm{o}T_\mathrm{i}}{U_\mathrm{i}T_\mathrm{o}^\mathrm{OA}}\right)^{-1}
 \right)$.
Define a function of $\eta$ as $f(\eta) = \eta T_\mathrm{i} + (1-\eta) T_\mathrm{o}^\mathrm{OA}=(T_\mathrm{i}-T_\mathrm{o}^\mathrm{OA})\eta + T_\mathrm{o}^\mathrm{OA}$. Since $T_\mathrm{i}>T_\mathrm{o}^\mathrm{OA}$, $f(\eta)$ monotonically increases with $\eta$. Thus, $\eta^\ast$ is the maximum $\eta \in \mathcal{Q}$, which yields (\ref{eq:Prop1}). Moreover, since $\mathcal{Q}=\emptyset$ for $\frac{\Omega_\mathrm{h}U_\mathrm{i}}{T_\mathrm{i}} >
\min  \Big[ 1-\frac{\Omega_\mathrm{c}U_\mathrm{o}}{T_\mathrm{o}^\mathrm{OA}},~
\left(1+\varepsilon\frac{U_\mathrm{o}T_\mathrm{i}}{U_\mathrm{i}T_\mathrm{o}^\mathrm{OA}}\right)^{-1}
 \Big]$, $\eta^\ast$ is feasible for $\eta^\ast \geq \frac{\Omega_\mathrm{h}U_\mathrm{i}}{T_\mathrm{i}}$.
\end{proof}
\begin{rmk}
Shared access with $\eta^\ast=1$ and $\eta^\ast=\rho_\mathrm{i}=\frac{U_i}{U_i+U_o}$ is closed and open access, respectively.
The QoS parameter $\varepsilon \in [0,1]$ determines the priority of home users relative to cellular users with $\varepsilon=1$ ensuring identical throughput to home and cellular users. In (\ref{eq:Prop1}), increasing $\varepsilon$ reduces $\eta^\ast$ and allocates more time-slots to cellular users. This indicates that shared access provides lower network throughput than open access when $\varepsilon$ is set high, e.g. such that $\eta^\ast < \frac{U_i}{U_i+U_o}$.
\end{rmk}
Fig. \ref{fig:TputN-HA} compares the network throughput for different femtocell access schemes, where the system parameters in Table \ref{table_1} are adopted. We assume the QoS requirement $\Omega_\mathrm{c}=0.01$ and $\Omega_\mathrm{h}=0.1$ respectively corresponding to 50 kbps and 500 kbps for 5 MHz bandwidth. Since $R_\mathrm{f}=R_\mathrm{i}$ at $D=D_\mathrm{th}$, we obtain the distance $D_\mathrm{th}$ of 130m by substituting $R_\mathrm{i}=20$ into (\ref{eq:Rf}). For $D>D_\mathrm{th}$, the throughput of shared access increases with decreasing $\varepsilon$. Considering lower $\varepsilon$ results in higher $\eta$, this indicates that increasing home user throughput $\eta T_\mathrm{i}$ counteracts the effects of decreasing cellular user throughput $(1-\eta) T_\mathrm{o}^\mathrm{OA}$. Moreover, this implies that the shared access with higher $\varepsilon$ (more time-slot allocation to cellular users) provides lower throughput than open access as shown in the result for $\varepsilon=0.1$.
For $D>D_\mathrm{th}$, closed access always provides higher throughput than shared access because shared access with $\eta=1$, which does not satisfy the QoS requirement, is the same as closed access. For $D\leq D_\mathrm{th}$, shared access obtains the same throughput as open access regardless of $\varepsilon$. The reason is that like open access, the shared access with time-slot allocation allows access from all neighboring cellular users located in the zone $\mathcal{F}_\mathrm{o}$. Note that the shared access with appropriate value of $\varepsilon$ achieves higher (at $D>D_\mathrm{th}$) or equal (at $D \leq D_\mathrm{th}$) network throughput than open access. We summarize these observations in Table \ref{table_2}.
\section{Conclusion}
\label{sec:conclusion}
The overall contribution of this paper is a new analytical framework for evaluating throughput tradeoffs regarding femtocell access schemes in downlink two-tier femtocell networks. The framework quantifies femtocell-site-specific ``loud neighbor'' effects and can be used to compare other techniques e.g.  power control, spectrum allocation, and MIMO. Our results show that unlike the uplink results in \cite{FemtoAccess3}, the preferred access scheme for home and cellular users is incompatible. In particular, closed access provides higher throughput for home users and lower throughput for neighboring cellular users; vice versa with open access. As a compromise, we suggest {\it shared access} where femtocells choose a time-slot ratio for their home and neighboring cellular users to maximize the network throughput subject to a network-wide QoS requirement. For femtocells within the outer area, shared access achieves higher network throughput than open access while satisfying the QoS of both home and cellular users. These results motivate shared access - i.e. open access, but with limits - in femtocell-enhanced cellular networks with universal frequency reuse.
%\appendices
\appendix[]
\subsection{Proof of Lemma 2}\label{sec:PfLem2}
In closed access the user SIR in (\ref{eq:SIR_FoCA}) is rewritten as $\gamma(R) = \frac{g_0}{K (I_1 +I_2)}$,
%\begin{equation}\label{eq:Lem2-1}
%\gamma(R) = \frac{g_0}{K (I_1 +I_2)},
%\end{equation}
where $K=\frac{P_\mathrm{f} L }{P_\mathrm{c} D^{-\alpha}}$ and $I_1=h_0 R^{-\alpha}$ and $I_2=\sum_{j \in \Phi\backslash{A_0}}  h_{j} |X_{j}|^{-\alpha}$. Then, the complementary cumulative distribution function (CCDF) of the user SIR at distance $R$ from the FAP is given as
\begin{eqnarray}\label{eq:Lem2-2}
\mathbb{P}[\gamma(R) \geq \Gamma]=\mathbb{P}[g_0 \geq \Gamma K(I_1+I_2)] %\cr
\mathop =\limits^{\left( a \right)}\int_0^\infty e^{-\Gamma K s} \mathrm{dPr}(I_1+I_2\leq s) %\cr
\mathop =\limits^{\left( b \right)}\mathcal{L}_{I_1}(\Gamma K )\mathcal{L}_{I_2}(\Gamma K),
\end{eqnarray}
where $(a)$ follows because the CCDF of exponential $g_0$ with unit mean is given as $\mathbb{P}[g_0>t]=e^{-t}$, and $(b)$ is given from \cite[Lemma 3.1]{Baccelli}. Here $\mathcal{L}_{I_1}(\Gamma K)$ is the Laplace transform of $I_1$ (exponential random variable scaled by $R^{-\alpha}$), which is given as
\begin{eqnarray}\label{eq:Lem2-3}
\mathcal{L}_{I_1}(\Gamma K)=\int_0^\infty e^{-\Gamma K s} f_{I_1}(s)\mathrm{d}s %\cr
=R^\alpha\int_0^\infty e^{-\Gamma K s} e^{-sR^\alpha} \mathrm{d}s %\cr
=\frac{1}{\Gamma K R^{-\alpha}+1},
\end{eqnarray}
Moreover, $\mathcal{L}_{I_2}(\Gamma K)$ is the Laplace transform of the Poisson shot-noise process $I_2$. For exponential $h_{j}$ with unit mean, $\mathcal{L}_{I_2}(\Gamma K)$ is given by \cite{Baccelli}
\begin{equation}\label{eq:Lem2-4}
\mathcal{L}_{I_2}(\Gamma K )=e^{-\lambda C_\alpha (K\Gamma)^{2/\alpha} },
\end{equation}
where $C_\alpha=\frac{2\pi^2}{\alpha}\csc(\frac{2\pi}{\alpha})$.
Thus, (\ref{eq:Lem2-2}) is simplifies to
\begin{equation}\label{eq:Lem2-5}
\mathbb{P}[\gamma(R) \geq \Gamma ]=\frac{e^{-\lambda C_\alpha (K\Gamma)^{2/\alpha} }}{\Gamma K R^{-\alpha}+1}.
\end{equation}
The cellular users are uniformly located at the zone $\mathcal{F}_\mathrm{o}$ that a circular annulus with outer radius $R_\mathrm{f}$ and inner radius $R_\mathrm{i}$. Then, probability density function (PDF) of the distance $R$ is $f_R(r|R_\mathrm{i}\leq R \leq R_\mathrm{f})=\frac{2r}{R_\mathrm{f}^2-R_\mathrm{i}^2}$. The spatially averaged SIR distribution over $\mathcal{F}_\mathrm{o}$ is given as
\begin{align}\label{eq:Lem2-6}
S_{\mathrm{o}}^{\mathrm{CA}}(\Gamma)=\mathbb{E}_R \left[\mathbb{P}[\gamma(R) \leq \Gamma ]|R_\mathrm{i}\leq R \leq R_\mathrm{f}\right]
%&=&1-\int_{R_\mathrm{i}}^{R_\mathrm{f}} \mathbb{P}[\gamma(R) \geq \Gamma ] f_R(r|R_\mathrm{i}\leq R \leq R_\mathrm{f}) \mathrm{d}R\cr
=1-\frac{2e^{-\lambda C_\alpha (K\Gamma)^{2/\alpha}}}{R_\mathrm{f}^2-R_\mathrm{i}^2}
\int_{R_\mathrm{i}}^{R_\mathrm{f}} \frac{R}{\Gamma K R^{-\alpha}+1} \mathrm{d}R.%\cr
%&\mathop =\limits^{\left( a \right)}
%1-\frac{e^{-\lambda C_\alpha (K\Gamma)^{2/\alpha}}}{R_\mathrm{f}^2-R_\mathrm{i}^2}\Big(R_\mathrm{f}^2-R_\mathrm{i}^2+R_\mathrm{i}^2
%~_2F_1\left[\tfrac{2}{\alpha},1;1+\tfrac{2}{\alpha};
%-\tfrac{R_\mathrm{i}^\alpha}{K\Gamma}\right]%\cr
%-R_\mathrm{f}^2~_2F_1\left[\tfrac{2}{\alpha},1;1+\tfrac{2}{\alpha};
%-\tfrac{R_\mathrm{f}^\alpha}{K\Gamma}\right]\Big),
\end{align}
Desired result (\ref{eq:S-CA}) is obtained by further calculating (\ref{eq:Lem2-6}) with the following integration formula \cite{HyperGeo}
\begin{equation}\label{eq:Lem2-7}
\int\frac{t}{at^{-\alpha}+1}\mathrm{d}t=
\frac{1}{2}t^2\left(1-~_2F_1\left[\frac{2}{\alpha},
1;1+\frac{2}{\alpha};-\frac{t^\alpha}{a}\right]\right).
\end{equation}

Next, in open access, the user SIR in (\ref{eq:SIR_FoOA}) is rewritten as $\gamma(R) = \frac{h_0}{R^{\alpha}(I_3 + I_4)}$,
%\begin{equation}\label{eq:Lem2-7}
%\gamma(R) = \frac{h_0}{R^{\alpha}(I_3 + I_4)},
%\end{equation}
where $I_3=K^{-1} g_0$ and $I_4=\sum_{j \in \Phi\backslash{A_0}}  h_{j} |X_{j}|^{-\alpha}$. Using the same approach as in (\ref{eq:Lem2-2}), CCDF of the user SIR at distance $R$ from the FAP is given as
\begin{eqnarray}\label{eq:Lem2-8}
\mathbb{P}[\gamma(R) \geq \Gamma]=\mathbb{P}[h_0 > \Gamma R^{\alpha}(I_3+I_4)] %\cr
=\mathcal{L}_{I_3}(\Gamma R^{\alpha} )\mathcal{L}_{I_4}(\Gamma R^{\alpha}).
\end{eqnarray}
As $I_3$ is the exponential random variable scaled by $K^{-1}$, from (\ref{eq:Lem2-3}) we obtain $\mathcal{L}_{I_3}(\Gamma R^{\alpha} )=1/(\Gamma K^{-1} R^{\alpha} + 1)$. As the Poisson shot-noise process $I_4$ is equal to $I_2$, we get $\mathcal{L}_{I_4}(\Gamma R^{\alpha} )=e^{-\lambda C_\alpha \Gamma^{2/\alpha}R^2 }$. Thus, we get
\begin{equation}\label{eq:Lem2-9}
\mathbb{P}[\gamma(R) \geq \Gamma ]=\frac{e^{-\lambda C_\alpha \Gamma^{2/\alpha}R^2 }}{\Gamma K^{-1} R^{\alpha} + 1}.
\end{equation}
For open access, from (\ref{eq:Lem2-6}) and (\ref{eq:Lem2-9}) the spatially averaged SIR distribution over $\mathcal{F}_\mathrm{o}$ is given as
\begin{align}\label{eq:Lem2-10}
S_{\mathrm{o}}^{\mathrm{OA}}(\Gamma)&=\mathbb{E}_R \left[\mathbb{P}[\gamma(R) \leq \Gamma ]|R_\mathrm{i}\leq R \leq R_\mathrm{f}\right] \cr
&=1-\frac{1}{R_\mathrm{f}^2-R_\mathrm{i}^2}
\int_{R_\mathrm{i}}^{R_\mathrm{f}} \frac{2R}{\Gamma K^{-1} R^{\alpha}+1} e^{-\lambda C_\alpha \Gamma^{2/\alpha}R^2 }\mathrm{d}R.
\end{align}
Here, by using substitution $R^2=r$, we obtain $\mathrm{d}r=2R\mathrm{d}R$ and (\ref{eq:S-OA}).

In particular, for $\alpha=4$ and by using substitution $R^2=r$, (\ref{eq:Lem2-10}) is rewritten as
\begin{eqnarray}\label{eq:Lem2-11}
S_{\mathrm{o}}^{\mathrm{OA}}(\Gamma)=1-\frac{1}{R_\mathrm{f}^2-R_\mathrm{i}^2}
\int_{R_i^2}^{R_\mathrm{f}^2} \frac{e^{-r\lambda C_\alpha \sqrt{\Gamma} }}{\Gamma K^{-1} r^2+1} \mathrm{d}r %\cr
= 1-\frac{B(R_\mathrm{f}^2)-B(R_\mathrm{i}^2)}{R_\mathrm{f}^2-R_\mathrm{i}^2},
\end{eqnarray}
with
\begin{align}\label{eq:Lem2-12}
B(x) = \int_{0}^{x}\frac{e^{yt}}{zt^2+1}\mathrm{d}t
%\cr
&\mathop =\limits^{\left( a \right)}
-\tfrac{i}{2\sqrt{z}}
\Big[e^{iy/\sqrt{z}}
\{E_i\left(xy-iy/\sqrt{z}\right)-E_i\left(-iy/\sqrt{z}\right)\}\cr
&~+e^{-iy/\sqrt{z}}\{E_i\left(iy/\sqrt{z}\right)
-E_i\left(xy+iy/\sqrt{z}\right)\}\Big]\cr
&\mathop =\limits^{\left( b \right)}
\tfrac{1}{\sqrt{z}}\Big[\left(-\mathrm{Re}\{\mathrm{Ei}(i y/\sqrt{z})\}+\mathrm{Re}
\{\mathrm{Ei}(xy+iy/\sqrt{z})\}\right)\sin(y/\sqrt{z})\cr
&~+\left(\mathrm{Im}\{\mathrm{Ei}(i y/\sqrt{z})\}-
\mathrm{Im}\{\mathrm{Ei}(xy+iy/\sqrt{z})\}\right)\cos(y/\sqrt{z})\Big],
\end{align}
where $y=-\lambda C_\alpha \sqrt{\Gamma} ,~ z=\Gamma K^{-1}$. $(a)$ follows from the integral formula in \cite{Integral2} and $(b)$ is given on the mirror symmetry of Exponential integral function, i.e. $\mathrm{Ei}(\bar{z})=\overline{\mathrm{Ei}(z)}$. Combining (\ref{eq:Lem2-11}) and (\ref{eq:Lem2-12}) gives the desired result (\ref{eq:S-OA1}).

\subsection{Proof of Lemma 3}\label{sec:PfLem4}
The home user SIR in (\ref{eq:SIR_Fi}) is rewritten as $\gamma(R) = \frac{h_0}{R^{\beta}(I_1 + I_2)}$,
%\begin{equation}\label{eq:Lem4-1}
%\gamma(R) = \frac{h_0}{R^{\beta}(I_1 + I_2)},
%\end{equation}
where $I_1=K g_0$ with $K=\frac{P_\mathrm{c} L}{P_\mathrm{f} D^{\alpha}}$ and $I_2=L^2\sum_{j \in \Phi\backslash{A_0}}  h_{j} |X_{j}|^{-\alpha}$. Using the way to obtain (\ref{eq:Lem2-9}), CCDF of the user SIR at distance $R$ from the FAP is given as
\begin{eqnarray}
\mathbb{P}[\gamma(R) \geq \Gamma ]%\cr\label{eq:Lem4-2-0}
=\mathcal{L}_{I_1}(\Gamma R^{\beta} )\mathcal{L}_{I_2}(\Gamma R^{\beta})%\\
=
\frac{e^{-\lambda C_\alpha (L^2 \Gamma)^{2/\alpha}R^{2\beta/\alpha} }}{\Gamma K R^{-\beta}+1},\label{eq:Lem4-2}
\end{eqnarray}
Assuming the home users are uniformly located in $\mathcal{F}_\mathrm{i}$, PDF of $R$ is $f_{R}(r|0\leq R \leq R_\mathrm{i})=\frac{2r}{R_\mathrm{i}^2}$. Thus, the spatially averaged SIR distribution of the home users is given as
\begin{eqnarray}\label{eq:Lem4-3}
S_{\mathrm{i}}(\Gamma)=\mathbb{E}_R \left[\mathbb{P}[\gamma(R) \leq \Gamma ]|0\leq R \leq R_\mathrm{i}\right] %\cr
=1-\frac{2}{R_\mathrm{i}^2}
\int_{0}^{R_\mathrm{i}} \frac{R \cdot e^{-\lambda C_\alpha (L^2 \Gamma)^{2/\alpha}R^{2\beta/\alpha}}}{\Gamma K R^{\beta}+1} \mathrm{d}R,
\end{eqnarray}
which proves (\ref{eq:S-F}).

In particular, for $\alpha=4$ and $\beta=2$, (\ref{eq:Lem4-3}) is rewritten as
\begin{eqnarray}\label{eq:Lem4-4}
S_{\mathrm{i}}(\Gamma)=1-\frac{2}{R_\mathrm{i}^2}
\int_{0}^{R_\mathrm{i}} \frac{R \cdot e^{-\lambda C_\alpha L \sqrt{\Gamma}R}}{\Gamma K R^2+1} \mathrm{d}R %\cr
= 1-\frac{H(R_\mathrm{i})}{R_\mathrm{i}^2}
\end{eqnarray}
with
\begin{align}\label{eq:Lem4-5}
H(x) = \int_{0}^{x}\frac{2te^{yt}}{zt^2+1}\mathrm{d}t
%\cr
&\mathop =\limits^{\left( a \right)}
\tfrac{1}{z}
\Big[e^{iy/\sqrt{z}}
\{E_i\left(xy-iy/\sqrt{z}\right)-E_i\left(-iy/\sqrt{z}\right)\}\cr
&~+e^{-iy/\sqrt{z}}\{E_i\left(xy+iy/\sqrt{z}\right)
-E_i\left(iy/\sqrt{z}\right)\}\Big]\cr
&\mathop =\limits^{\left( b \right)}
\tfrac{2}{z}\Big[\left(\mathrm{Re}\{\mathrm{Ei}(xy+iy/\sqrt{z})\}
-\mathrm{Re}\{\mathrm{Ei}(iy/\sqrt{z})\}\right)\cos(y/\sqrt{z})\cr
&~+\left(\mathrm{Im}\{\mathrm{Ei}(xy+iy/\sqrt{z})\}-\mathrm{Im}\{\mathrm{Ei}(iy/\sqrt{z})\}
\right)\sin(y/\sqrt{z})\Big],
\end{align}
where $y=-\lambda C_\alpha L \sqrt{\Gamma}, ~ z=K\Gamma$. $(a)$ follows from the integral formula in \cite{Integral2} and $(b)$ is given on the mirror symmetry of Exponential integral function. Combining (\ref{eq:Lem4-4}) and (\ref{eq:Lem4-5}) gives the desired result (\ref{eq:S-F1}).

Moreover, for the path loss exponents $\alpha=\beta=4$, by substitution $R^2=r$, (\ref{eq:Lem4-3}) is rewritten as
\begin{eqnarray}\label{eq:Lem4-6}
S_{\mathrm{i}}(\Gamma)=1-\frac{1}{R_\mathrm{i}^2}
\int_{0}^{R_\mathrm{i}^2} \frac{ e^{-r\lambda C_\alpha L \sqrt{\Gamma}}}{\Gamma K r^2+1} \mathrm{d}r %\cr
= 1-\frac{B(R_\mathrm{i}^2)}{R_\mathrm{i}^2},
\end{eqnarray}
where $B(x)$ is given from (\ref{eq:Lem2-12}). This gives the desired spatially averaged SIR distribution in (\ref{eq:S-F2}).

\subsection{Proof of Lemma 4}\label{sec:PfLem5}
The user SIR in (\ref{eq:SIR_FoCA}) is rewritten as $\gamma(R) = \frac{g_0}{K (I_1 +I_2)}$,
%\begin{equation}\label{eq:Lem5-1}
%\gamma(R) = \frac{g_0}{K (I_1 +I_2)},
%\end{equation}
where $K=\frac{P_\mathrm{f} D^{\alpha} }{P_\mathrm{c} L}$ and $I_1=h_0 R^{-\beta}$ and $I_2=L^2 \sum_{j \in \Phi\backslash{A_0}}  h_{j} |X_{j}|^{-\alpha}$. Using the same approach as in (\ref{eq:Lem2-2}), CCDF of the user SIR at distance $R$ from the FAP is given as
\begin{eqnarray}\label{eq:Lem5-2}
\mathbb{P}[\gamma(R) \geq \Gamma ]
=\mathcal{L}_{I_1}(\Gamma K )\mathcal{L}_{I_2}(\Gamma K) %\cr
=\frac{e^{-\lambda C_\alpha (L^2 K  \Gamma)^{2/\alpha} }}{\Gamma K R^{-\beta}+1}
\end{eqnarray}
The home users connected to the central MBS are uniformly located at the zone $\mathcal{F}_\mathrm{b}$ that a circular annulus with outer radius $R_\mathrm{i}$ and inner radius $R_\mathrm{f}$. Then, PDF of $R$ is $f_R(r|R_\mathrm{f}\leq R \leq R_\mathrm{i})=\frac{2r}{R_\mathrm{i}^2-R_\mathrm{f}^2}$. The spatially averaged SIR distribution over $\mathcal{F}_\mathrm{b}$ is given as
\begin{eqnarray}\label{eq:Lem5-6}
S_{\mathrm{b}}(\Gamma)=\mathbb{E}_R \left[\mathbb{P}[\gamma(R) \leq \Gamma]|R_\mathrm{f}\leq R \leq R_\mathrm{i}\right]
=1-\frac{2e^{-\lambda C_\alpha (L^2 K \Gamma)^{2/\alpha}}}{R_\mathrm{i}^2-R_\mathrm{f}^2}
\int_{R_\mathrm{f}}^{R_\mathrm{i}} \frac{R}{\Gamma K R^{-\beta}+1} \mathrm{d}R.
\end{eqnarray}
Applying (\ref{eq:Lem2-7}) to (\ref{eq:Lem5-6}), we proves (\ref{eq:S-Fb}).

\subsection{Proof of Theorem 1}\label{sec:PfThm1}
For a reference FAP at $D\leq D_\mathrm{th}$, the home users in $\mathcal{F}_\mathrm{a}$ connect to the FAP, while the remaining home users in $\mathcal{F}_\mathrm{b}$ communicate to the $B_0$. Thus, the average sum throughput of the home users is given as $T_\mathrm{h}^\mathrm{CA} = T_1 + T_2$,
%\begin{equation}\label{eq:T-h-CA1}
%T_\mathrm{h}^\mathrm{CA} = T_1 + T_2, ~~ D\leq D_\mathrm{th}
%\end{equation}
where $T_1$ and $T_2$ is the average sum throughput of the home users in $\mathcal{F}_\mathrm{a}$ and $\mathcal{F}_\mathrm{b}$, respectively.
Since the FAP supports the home users in $\mathcal{F}_\mathrm{a}$ only, we obtain $T_1 = T_\mathrm{a}$,
%\begin{equation}\label{eq:T-h-CA2}
%T_1 = T_\mathrm{a},
%\end{equation}
On the other hand, since the MBS transmits data to cellular users as well as the remaining $U_\mathrm{b}$ home users in $\mathcal{F}_\mathrm{b}$, $T_2$ is given as $T_2 = \rho_\mathrm{b}^{\mathrm{CA}} T_\mathrm{b}$, and thus we get
\begin{equation}\label{eq:T-h-CA1}
T_\mathrm{h}^\mathrm{CA} = T_\mathrm{a} + \rho_\mathrm{b}^{\mathrm{CA}} T_\mathrm{b}, ~~ D\leq D_\mathrm{th}.
\end{equation}
%\begin{equation}\label{eq:T-h-CA3}
%T_2 = \rho_\mathrm{b}^{\mathrm{CA}} T_\mathrm{b},
%\end{equation}
Here $\rho_\mathrm{b}^{\mathrm{CA}}$ is the fraction of time-slot dedicated to the $U_\mathrm{b}$ home users among all $\overline{U}_\mathrm{m}^\mathrm{CA}$ users supported by the MBS with a RR scheduler, which is given as
\begin{eqnarray}
\rho_\mathrm{b}^{\mathrm{CA}} = U_\mathrm{b} / \overline{U}_\mathrm{m}^\mathrm{CA} = U_\mathrm{b} / (U_\mathrm{c} + N_\mathrm{f1} \overline{U}_\mathrm{b}).
\label{eq:rho-CA}
\end{eqnarray}
where $U_\mathrm{c}$ is the number of cellular users, and $N_\mathrm{f1}$ is the number of femtocells with $D\leq D_\mathrm{th}$.
Moreover, $U_\mathrm{b}$ is given as
\begin{eqnarray}
U_\mathrm{b} \mathop =\limits^{\left( a \right)} U_\mathrm{h} \left(1-\left(\tfrac{R_\mathrm{f}}{R_\mathrm{i}}\right)^2\right) %\cr
\mathop =\limits^{\left( b \right)}
U_\mathrm{h} \left(1-\tfrac{\kappa^{2/\alpha}D^2}{(\kappa^{2/\alpha}-1)^2 R_\mathrm{i}^2} \right),
\label{eq:Ub}
\end{eqnarray}
where (a) is given on the uniform distribution assumption of home users and (b) follows from (\ref{eq:Rf}). In (\ref{eq:rho-CA}), $\overline{U}_\mathrm{b}$ denotes the average number of users in $\mathcal{F}_\mathrm{b}$, which given as
\begin{eqnarray}
\overline{U}_\mathrm{b} = \mathbb{E}[U_\mathrm{b}]%\cr
= U_\mathrm{h} \left(1-\tfrac{\kappa^{2/\alpha}}{(\kappa^{2/\alpha}-1)^2 R_\mathrm{i}^2}\mathbb{E}[D^2] \right)%\cr
\mathop =\limits^{\left( a \right)}
U_\mathrm{h} \left(1-\tfrac{\kappa^{2/\alpha}D_\mathrm{th}^2}{2(\kappa^{2/\alpha}-1)^2 R_\mathrm{i}^2} \right),
\label{eq:AvgUb}
\end{eqnarray}
where (a) follows from $\mathbb{E}[D^2]=\int_{0}^{D_\mathrm{th}} D^2 \left(\frac{2D}{D_\mathrm{th}^2}\right)\mathrm{d}D=D_\mathrm{th}^2/2$ for $D\leq D_\mathrm{th}$.
Combining (\ref{eq:rho-CA}), (\ref{eq:Ub}), and (\ref{eq:AvgUb}) gives the desired result in (\ref{eq:thm1rho-CA}).

Next, since a reference FAP with $D>D_\mathrm{th}$ supports $U_\mathrm{i}$ home users in the zone $\mathcal{F}_\mathrm{i}$ only, the average sum throughput, $T_\mathrm{h}^\mathrm{CA}$, of the home users is equal to $T_\mathrm{i}$, which proves (\ref{eq:T-f-CA}) for $D>D_\mathrm{th}$.
For a reference FAP at $D > D_\mathrm{th}$, its neighboring cellular users in the zone $\mathcal{F}_\mathrm{o}$ connect to the central MBS in closed access. Since the MBS using TDMA transmits data to other cellular users as well as the neighboring cellular users in $\mathcal{F}_\mathrm{o}$, the average sum throughput of the neighboring cellular users is given as
\begin{equation}
T_\mathrm{c}^\mathrm{CA} = \rho_\mathrm{o} T_\mathrm{o}^\mathrm{CA},
\end{equation}
where $\rho_\mathrm{o}$ is the fraction of time-slot dedicated to the $U_\mathrm{o}$ neighboring cellular users in $\mathcal{F}_\mathrm{o}$ among all $\overline{U}_\mathrm{m}^\mathrm{CA}$ users supported by the MBS with a RR scheduler, which is given as
\begin{eqnarray}
\rho_\mathrm{o} = U_\mathrm{o} / \overline{U}_\mathrm{m}^\mathrm{CA} = U_\mathrm{o} / (U_\mathrm{c} + N_\mathrm{f1} \overline{U}_\mathrm{b}),
\label{eq:rho-o}
\end{eqnarray}
where $U_\mathrm{o}$ is given as
\begin{eqnarray}
U_\mathrm{o} \mathop =\limits^{\left( a \right)} U_\mathrm{c} \tfrac{R_\mathrm{f}^2-R_\mathrm{i}^2}{R_\mathrm{c}^2-N_\mathrm{f}R_\mathrm{i}^2} %\cr
\mathop =\limits^{\left( b \right)}
U_\mathrm{c} \frac{\kappa^{2/\alpha}D^2-(\kappa^{2/\alpha}-1)^2 R_\mathrm{i}^2}
{(\kappa^{2/\alpha}-1)^2(R_\mathrm{c}^2-N_\mathrm{f}R_\mathrm{i}^2)}.
\label{eq:Uo}
\end{eqnarray}
Here, (a) follows from the uniform distribution assumption of cellular users and (b) is given from (\ref{eq:Rf}). Combining (\ref{eq:AvgUb}), (\ref{eq:rho-o}), and (\ref{eq:Uo}) gives the desired result in (\ref{eq:thm1rho-o}).

\subsection{Proof of Theorem 2}\label{sec:PfThm2}
In open access, for a reference FAP at $D \leq D_\mathrm{th}$, the femtocell/macrocell access scenario of the home users in the zone $\mathcal{F}_\mathrm{o}$ and $\mathcal{F}_\mathrm{o}$  is the same as that in closed access. Thus, from (\ref{eq:T-h-CA1}) the average sum throughput of the home users is thus given
\begin{equation}\label{eq:T-h-OA1}
T_\mathrm{h}^\mathrm{OA} = T_\mathrm{a} + \rho_\mathrm{b}^{\mathrm{OA}} T_\mathrm{b}, ~~ D\leq D_\mathrm{th}
\end{equation}
where $\rho_\mathrm{b}^{\mathrm{OA}}$ is is the fraction of time-slot dedicated to the $U_\mathrm{b}$ home users among all $\overline{U}_\mathrm{m}^\mathrm{OA}$ users supported by the MBS with a RR scheduler, which is given as
\begin{eqnarray}
\rho_\mathrm{b}^{\mathrm{OA}} = U_\mathrm{b}/\overline{U}_\mathrm{m}^\mathrm{OA} = U_\mathrm{b} / (U_\mathrm{c} - N_\mathrm{f2} \overline{U}_\mathrm{o} + N_\mathrm{f1} \overline{U}_\mathrm{b})
\label{eq:rho-OA}
\end{eqnarray}
where $\overline{U}_\mathrm{m}^\mathrm{OA} = U_\mathrm{c} - N_\mathrm{f2} \overline{U}_\mathrm{o} + N_\mathrm{f1} \overline{U}_\mathrm{b}$ is the number of users served by the MBS in femtocell open access. Here $N_\mathrm{f2} \overline{U}_\mathrm{o}$ is the number of cellular users accessing to the FAP ($D>D_\mathrm{th}$) with open access. The average number of users in $\mathcal{F}_\mathrm{o}$, $\overline{U}_\mathrm{o}$, is given by
\begin{eqnarray}
\overline{U}_\mathrm{o} = \mathbb{E}[U_\mathrm{o}]%\cr
\mathop =\limits^{\left( a \right)}
U_\mathrm{c} \frac{\kappa^{2/\alpha}\mathbb{E}[D^2]-(\kappa^{2/\alpha}-1)^2 R_\mathrm{i}^2}
{(\kappa^{2/\alpha}-1)^2(R_\mathrm{c}^2-N_\mathrm{f}R_\mathrm{i}^2)}%\cr
\mathop =\limits^{\left( b \right)}
U_\mathrm{c} \frac{\kappa^{2/\alpha}(R_\mathrm{c}^2+D_\mathrm{th}^2)-2(\kappa^{2/\alpha}-1)^2 R_\mathrm{i}^2}
{2(\kappa^{2/\alpha}-1)^2(R_\mathrm{c}^2-N_\mathrm{f}R_\mathrm{i}^2)},
\label{eq:AvgUo}
\end{eqnarray}
where (a) is given from (\ref{eq:Uo}), and (b) follows from $\mathbb{E}[D^2]=\int_{D_\mathrm{th}}^{R_\mathrm{c}} D^2 \left(\frac{2D}{R_\mathrm{c}^2-D_\mathrm{th}^2}\right)\mathrm{d}D=(R_\mathrm{c}^2+D_\mathrm{th}^2)/2$ for $D > D_\mathrm{th}$. Combining (\ref{eq:Ub}), (\ref{eq:AvgUb}), (\ref{eq:rho-OA}), and (\ref{eq:AvgUo}) gives the desired result in (\ref{eq:thm2rho-OA}).

For a reference FAP at $D > D_\mathrm{th}$, since by using TDMA the MBS transmits data to the neighboring cellular users in $\mathcal{F}_\mathrm{o}$ as well as the home users in $\mathcal{F}_\mathrm{i}$, the average sum throughput of the home users is given as
\begin{eqnarray}
T_\mathrm{h}^\mathrm{OA} = \rho_\mathrm{i} T_\mathrm{i}, ~~
T_\mathrm{c}^\mathrm{OA} = (1-\rho_\mathrm{i}) T_\mathrm{o}^\mathrm{OA},
\end{eqnarray}
where $\rho_\mathrm{i}$ is the fraction of time slot dedicated to the home users in $\mathcal{F}_\mathrm{i}$ among home and cellular users supported from the FAP with a RR scheduler, which is given as
\begin{eqnarray}
\rho_\mathrm{i} = {U_\mathrm{i}} / (U_\mathrm{i} + U_\mathrm{o} )
\label{eq:rho-i},
\end{eqnarray}
Combining (\ref{eq:Uo}) and (\ref{eq:rho-i}) gives the desired result in (\ref{eq:thm2rho-i}).

\newpage
\begin{table}
\caption{Notation and simulation values}
\label{table_1}
\begin{center}
\begin{tabular}{c|l|c}
\hline
Symbol & Description & Sim. Value \\ \hline\hline
$\mathcal{F}_\mathrm{i}$ & Indoor area covered by the FAP at $D>D_\mathrm{th}$ (a disc with the radius $R_i$)& N/A\\
$\mathcal{F}_\mathrm{o}$ & Outdoor area covered by the FAP at $D>D_\mathrm{th}$ in open access or covered by the MBS & N/A\\
& in closed access (a circular annulus with inner radius $R_\mathrm{i}$ and outer radius $R_\mathrm{f}$) & \\
$\mathcal{F}_\mathrm{a}$ & Indoor area covered by the FAP at $D\leq D_\mathrm{th}$ (a disc with the radius $R_\mathrm{f}$)& N/A\\
$\mathcal{F}_\mathrm{b}$ &  Indoor area covered by the MBS (a circular annulus with inner radius $R_\mathrm{f}$ and & N/A
\\ &  outer radius $R_\mathrm{i}$ with respect to the FAP at $D\leq D_\mathrm{th}$)&\\
$D$ & Distance between FAP and central MBS & Not fixed\\
$D_\mathrm{th}$ & Threshold distance (Radius of inner region) & Not fixed \\
$D_\mathrm{c}$ & Distance between central MBS and homeuser (or neighboring cellular user) & Not fixed\\
$R$ & Distance between FAP and homeuser (or neighboring cellular user) & Not fixed\\
$R_\mathrm{f}$ & Femtocell radius & Not fixed\\
$R_\mathrm{c}$ & Macrocell radius &  500 m \\
$R_\mathrm{i}$ & Indoor (home) area radius & 20 m \\
$P_\mathrm{c}$ & Transmit power at macrocell & 43 dBm \cite{3GPPTS}\\
$P_\mathrm{f}$& Transmit power at femtocell & 13 dBm \cite{3GPPTS}\\
$\alpha$ & Outdoor path loss exponent & 4 \\ \
$\beta$ & Indoor path loss exponent & 4 \\
$L$ & Wall penetration loss & 0.5 (-3 dB) \\
$G$ & Shannon gap & 3 dB \\
$N$ & Number of discrete levels for M-ary modulation (M-QAM) & 8 \\
$\Omega_\mathrm{c}$ & Required minimum throughput of cellular user for hybrid access & 0.01 bps/Hz\\
$\Omega_\mathrm{h}$ & Required minimum throughput of home user for hybrid access & 0.1 bps/Hz\\ \hline
\end{tabular}
\end{center}
\end{table}

\begin{table}
\caption{Throughput comparison of open, closed, and shared access for FAP-MBS distance $D$, cellular user density $U_\mathrm{c}$, femtocell density $N_\mathrm{f}$}
\label{table_2}
\begin{center}
\begin{tabular}{c|c|c|c}
\hline
\multicolumn {2}{c|}{} & High $D$ and/or High $U_\mathrm{c}$ and/or Low $N_\mathrm{f}$ & Low $D$ and/or Low $U_\mathrm{c}$ and/or High $N_\mathrm{f}$ \\ \hline \hline
\multirow{2}{22mm}{Home user sum throughput} & inner region & Open = Shared $>$ Closed  & Open = Shared $\gg$ Closed \\ \cline{2-4}
& outer region & Closed $\gg$ Shared $\gg$ Open & Closed $>$ Shared $>$ Open \\ \hline
\multicolumn {2}{c|}{Cellular user sum throughput} & Open $\gg$ Shared $\gg$ Closed & Open $>$ Shared $>$ Closed \\ \hline
% &  &  \\ \hline
\multirow{3}{22mm} {Preferred access schemes } & home users & Closed access &   Closed access \\ \cline{2-4} & cellular users & Open access & Open access \\ \cline{2-4}
& home and cellular users & Shared access & Shared access \\ \hline
%& High $D$ and/or High $U_\mathrm{c}$ and/or Low $N_\mathrm{f}$ & Low $D$ and/or Low $U_\mathrm{c}$ and/or High $N_\mathrm{f}$ \\ \hline
%Home user sum throughput (inner region) & Open = Shard $>$ Closed & Open = Shard $>>$ Closed \\ \hline
%Home user sum throughput (outer region) & Closed $>>$ Shard $>>$ Open  & Closed $>$ Shard $>$ Open \\ \hline
%Cellular user sum throughput & Open $>>$ Shard $>>$ Closed & Open $>$ Shard $>$ Closed \\ \hline
%Preferred access for home and cellular users & more incompatible &   incompatible \\ \hline
\end{tabular}
\end{center}
\end{table}

\begin{figure}
\begin{center}
   \includegraphics[width=4.in]{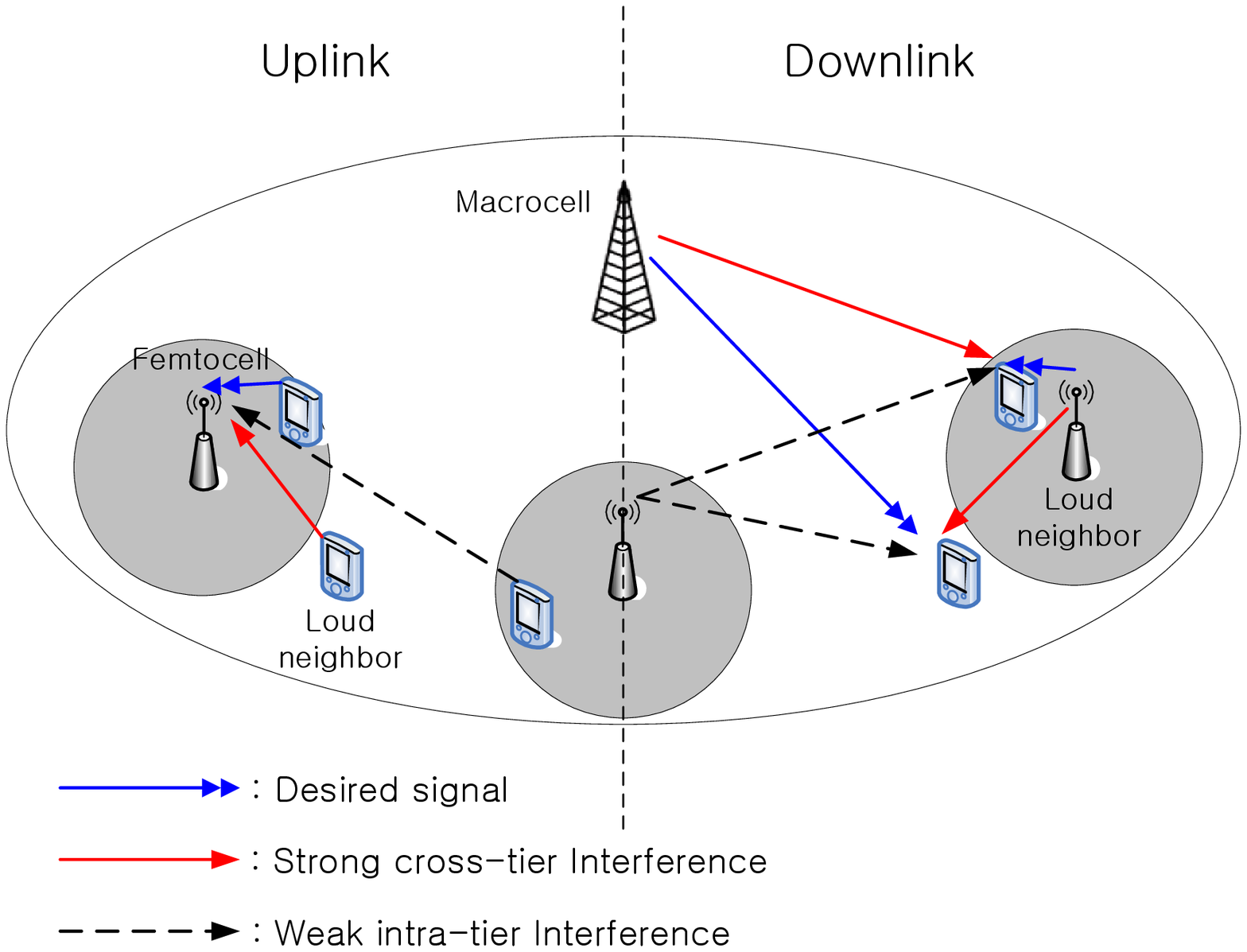}
    \caption{Loud neighbor effect in uplink and downlink two-tier femtocell networks}
    \label{fig:Loudneighbor}
\end{center}
\end{figure}

\begin{figure}
\begin{center}
   \includegraphics[width=4.in]{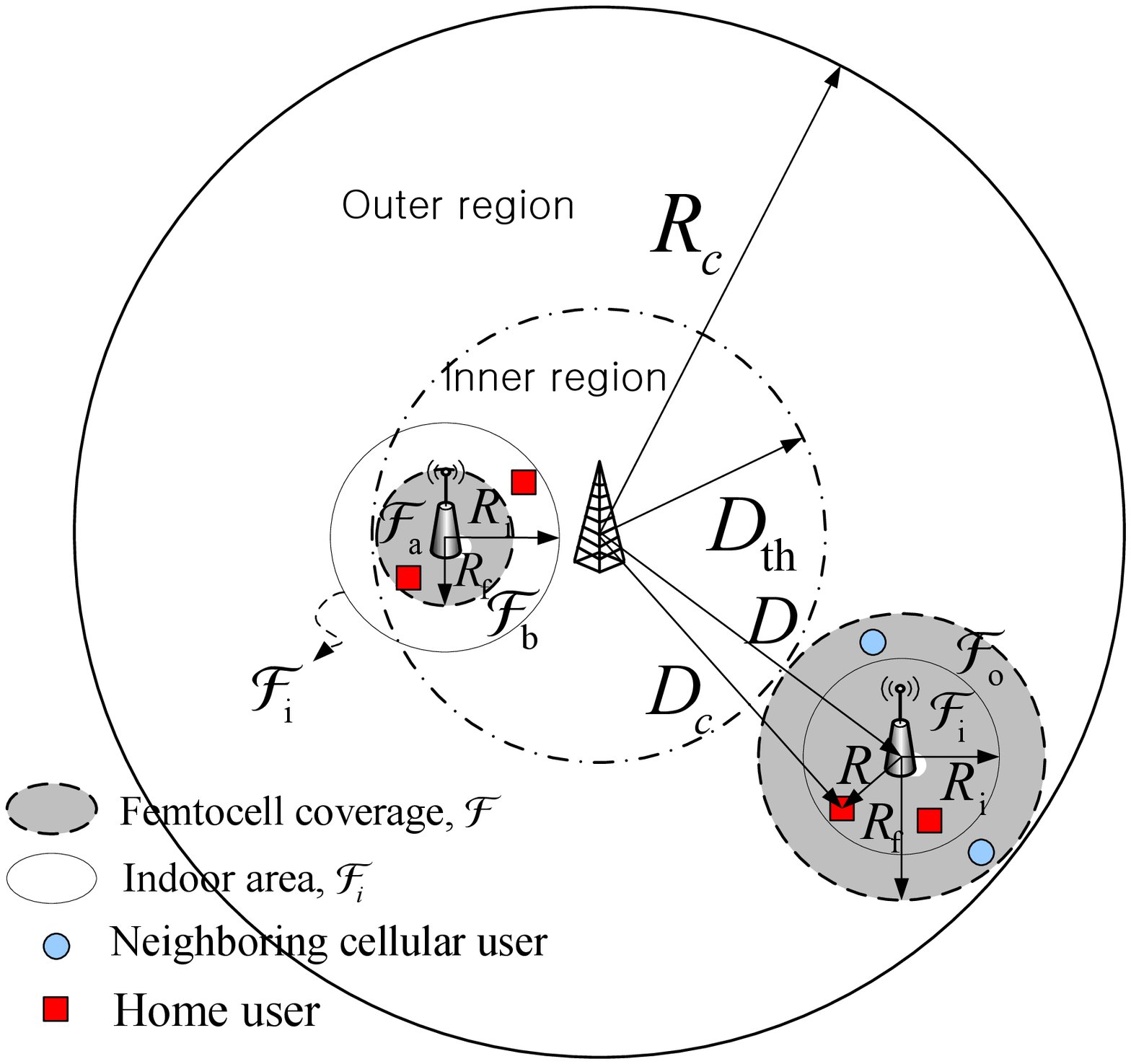}
    \caption{Femtocell coverage variation for the FAP-MBS distance $D$ and geometrical zone $\mathcal{F}_\mathrm{a}$, $\mathcal{F}_\mathrm{b}$, $\mathcal{F}_\mathrm{i}$, and $\mathcal{F}_\mathrm{o}$ in a two-tier femtocell networks with cochannel deployment. For $D>D_\mathrm{th}$, femtocell coverage $\mathcal{F}$ is larger than indoor area $\mathcal{F}_\mathrm{i}$, i.e., $R_\mathrm{f}>R_\mathrm{i}$.}
    \label{fig:coverage}
\end{center}
\end{figure}

\begin{figure}
\begin{center}
   \includegraphics[width=4.in]{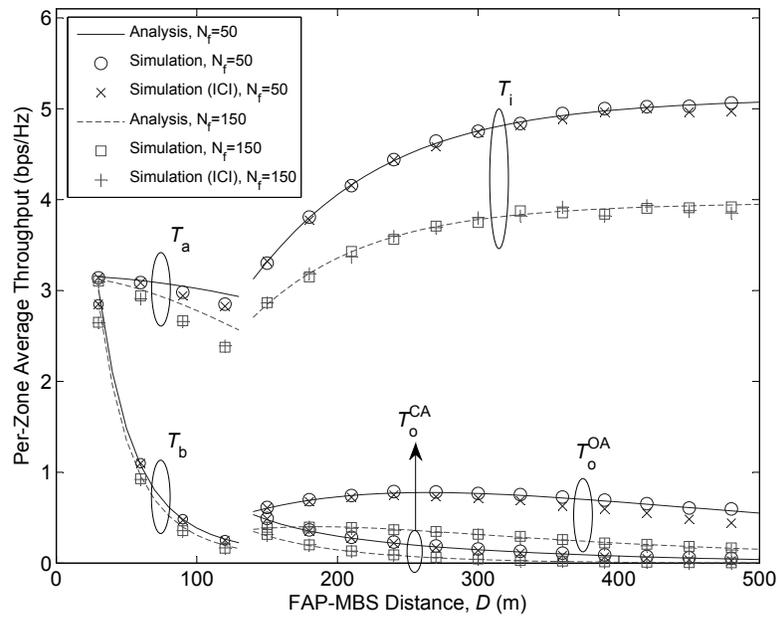}
    \caption{Theoretical and simulation results of spatially averaged throughput for the zone $\mathcal{F}_\mathrm{a}$, $\mathcal{F}_\mathrm{b}$, $\mathcal{F}_\mathrm{i}$, and $\mathcal{F}_\mathrm{o}$ ($D_\mathrm{th}=130 \textrm{m}$)}
    \label{fig:zoneTput}
\end{center}
\end{figure}

\begin{figure}[h]
\begin{center}
\subfigure[]{\epsfxsize=4in
\leavevmode\epsfbox{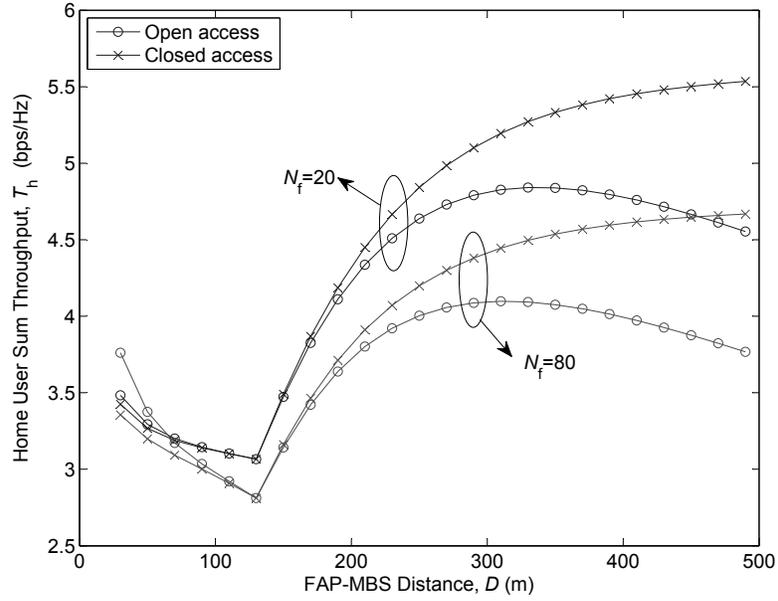}\label{fig:TputFwrtNf}}\\
\subfigure[]{\epsfxsize=4in
\leavevmode\epsfbox{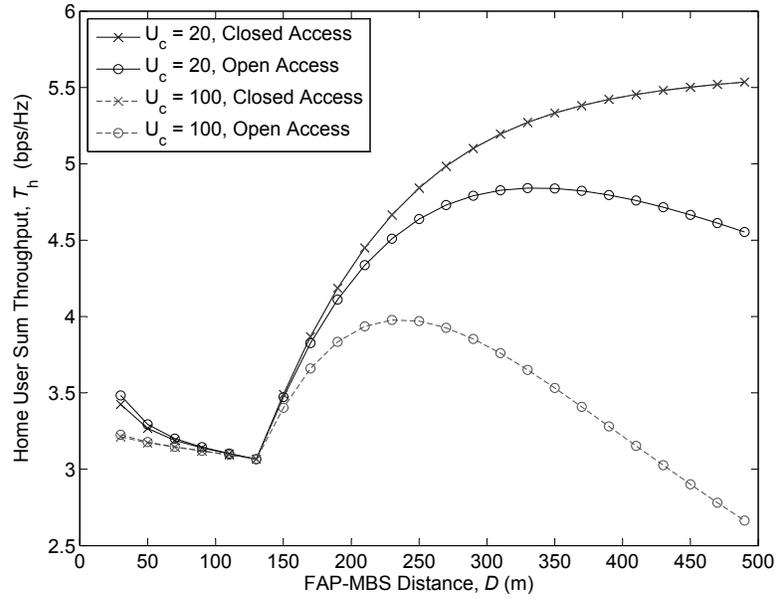}\label{fig:TputFwrtUc}}
\caption{Average sum throughput of the home users for different number of femtocells and cellular users ($D_\mathrm{th}=130 \textrm{m}$): (a) number of cellular users $U_\mathrm{c}=20$. (b) number of femtocells $N_\mathrm{f}=20$.} \label{fig:TputF}
\end{center}
\end{figure}

%\begin{figure}[h]
%\begin{center}
%\subfigure[]{\epsfxsize=4in
%\leavevmode\epsfbox{./Fig4-a.eps}\label{fig:TputC-OA}}\\
%\subfigure[]{\epsfxsize=4in
%\leavevmode\epsfbox{./Fig4-b.eps}\label{fig:TputC-CA}}
%\caption{Average sum throughput of the neighboring cellular users for different number of femtocells and cellular users ($D_\mathrm{th}=130 \textrm{m}$): (a) Open access (b) Closed access.} \label{fig:TputC}
%\end{center}
%\end{figure}

\begin{figure}
\begin{center}
   \includegraphics[width=4.in]{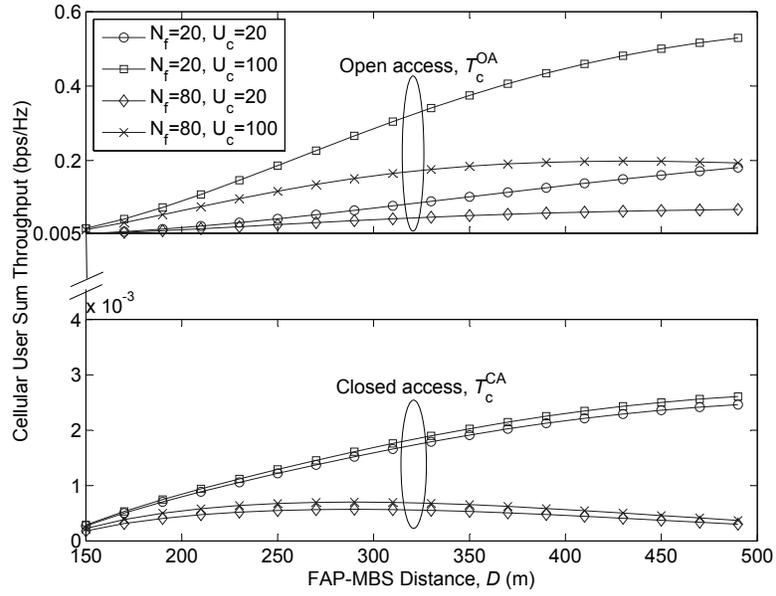}
    \caption{Average sum throughput of the neighboring cellular users for different number of femtocells and cellular users, ($D_\mathrm{th}=130 \textrm{m}$)}
    \label{fig:TputC}
\end{center}
\end{figure}

\begin{figure}
\begin{center}
   \includegraphics[width=4.in]{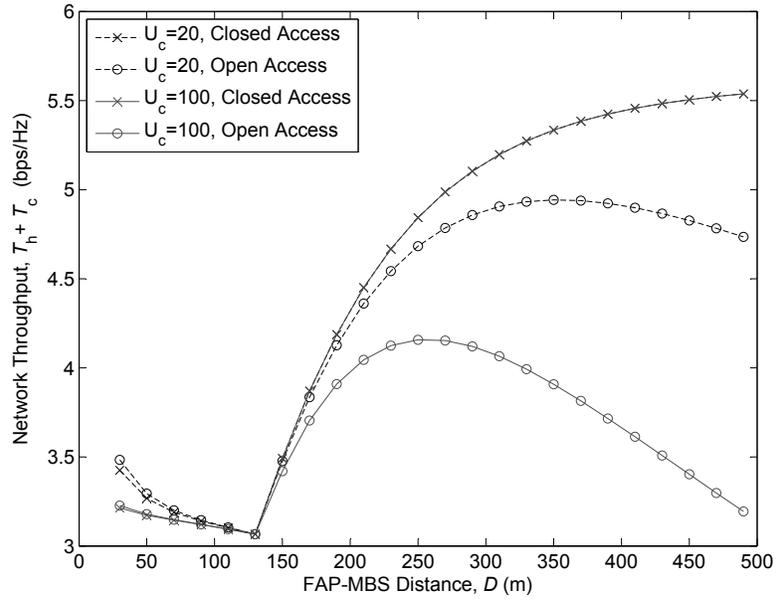}
    \caption{Network throughput for different number of cellular users, ($N_\mathrm{f}=20$ and $D_\mathrm{th}=130 \textrm{m}$)}
    \label{fig:TputN}
\end{center}
\end{figure}

\begin{figure}
\begin{center}
   \includegraphics[width=4.in]{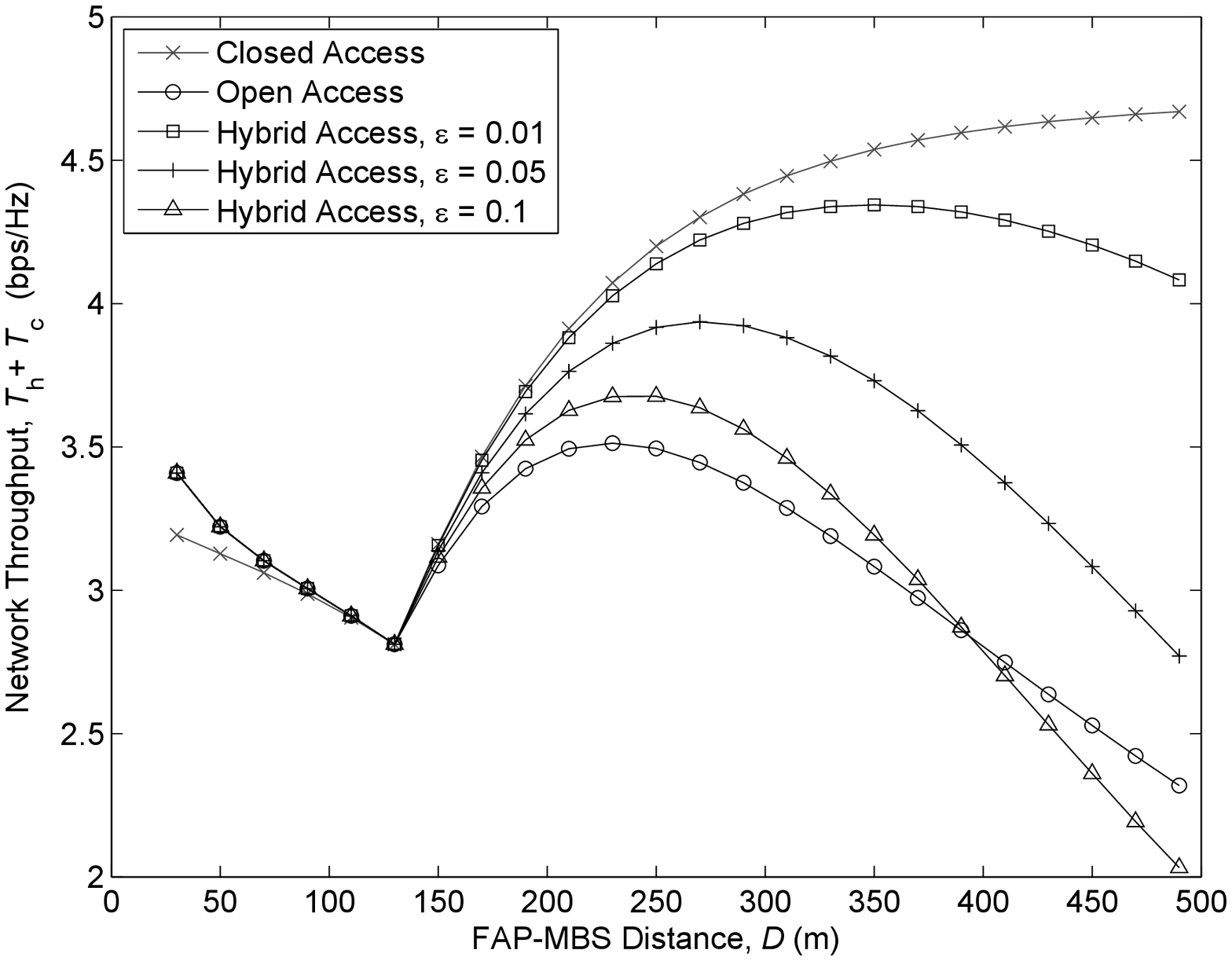}
    \caption{Network throughput for different femtocell access, ($N_\mathrm{f}=80$, $U_\mathrm{c}=100$, and $D_\mathrm{th}=130 \textrm{m}$)}
    \label{fig:TputN-HA}
\end{center}
\end{figure}

\begin{thebibliography}{99}% more than 9 --> 99 / less than 10 --> 9
\bibitem{FemtoIntro1}
H. Claussen, L. T. W. Ho, and L. G. Samuel, ``An Overview of the
Femtocell Concept,'' {\it Bell Labs Technical Journal}, vol. 13,
Issue 1, pp. 221-245, May 2008.

\bibitem{FemtoIntro2}
V. Chandrasekhar, J. G. Andrews, and A. Gatherer, ``Femtocell networks:
a survey,'' {\it IEEE Commun. Mag.}, vol. 46, no. 9, pp. 59-67, Sep. 2008.

\bibitem{FemtoIntro3}
``Femto forum femtocell business case whitepaper,'' White Paper, Signals Research Group, Femto Forum, June 2009.

\bibitem{FemtoPower1}
V. Chandrasekhar, J. G. Andrews, T. Muharemovic, Z. Shen, and A. Gatherer, ``Power control in two-tier femtocell networks,'' {\it IEEE
Trans. Wireless Commun.}, vol. 8, no. 8, pp. 4316-4328, August 2009.

\bibitem{FemtoPower2}
M. Yavuz, F. Meshkati, S. Nanda, A. Pokhariyal, N. Johnson, B. Raghothaman, and A. Richardson, ``Interference management and
performance analysis of UMTS/HSPA+ femtocells,'' {\it IEEE Commun. Mag.}, vol. 47, no. 9, pp. 102-109, Sep. 2009.

\bibitem{FemtoPower3}
H.-S. Jo, C. Mun, J. Moon, and J.-G. Yook, ``Interference Mitigation
Using Uplink Power Control for Two-Tier Femtocell Networks,'' {\it IEEE Trans. Wireless Commun.}, vol. 8, no. 10, pp. 4906-4910,
Oct. 2009.

\bibitem{FemtoPower4}
H.-S. Jo, C. Mun, J. Moon, and J.-G. Yook, ``Self-optimized Coverage
Coordination in Femtocell Networks,'' {\it accepted to IEEE Trans. Wireless Commun.}, Aug. 2010, [Online] Available at
http://arxiv.org/abs/0910.2168.

\bibitem{FemtoPower5}
V. Chandrasekhar, M. Kountouris and J. G. Andrews, ``Coverage in Multi-Antenna Two-Tier Networks'', {\it IEEE Trans. Wireless Commun.}, Vol. 8, No. 10, pp. 5314-5327, Oct. 2009.

\bibitem{FemtoFreq1}
I. Guvenc, M.-R. Jeong, F. Watanabe, and H. Inamura, ``A hybrid
frequency assignment for femtocells and coverage area analysis for cochannel
operation,'' {\it IEEE Commun. Lett.}, vol. 12, no. 12, Dec. 2008.

\bibitem{FemtoFreq2}
D. L$\acute{\mathrm{o}}$pez-P$\acute{\mathrm{e}}$rez, $\acute{\mathrm{A}}$. Lad$\acute{\mathrm{a}}$nyi, A. J$\ddot{\mathrm{u}}$ttner, and J. Zhang, `` OFDMA femtocells: A self-organizing approach for frequency assignment,''
in {\it Proc. IEEE PIMRC}, Sep. 2009, pp. 2202-2207.

\bibitem{FemtoFreq3}
V. Chandrasekhar and J. G. Andrews, ``Spectrum Allocation in Tiered Cellular Networks,'' {\it IEEE Trans. Commun.}, Vol. 57, No. 10, pp. 3059-3068, Oct. 2009.

\bibitem{FemtoCR1}
J. P. Torregoza, R. Enkhbat, and W.-J. Hwang, ``Joint power control, base station assignment, and channel assignment in cognitive femtocell networks,'' {\it EURASIP Journal of Wireless Communications and Networking}, vol. 2010, Article ID 285714, 14 pages, 2010. doi:10.1155/2010/285714

\bibitem{FemtoCR2}
S.-Y. Lien, C.-C. Tseng, K.-C. Chen, and C.-W. Su, ``Cognitive Radio Resource Management for QoS Guarantees in Autonomous Femtocell Networks,'' {\it to appear Proc. IEEE ICC 2010}, [Online] Available at http://santos.ee.ntu.edu.tw/Publication/03-14-05.pdf.

\bibitem{FemtoAccess1}
H. Claussen, ``Performance of macro-and co-channel femtocells in a hierarchical cell structure,'' {\it Proc. IEEE PIMRC}, Sep. 2007, pp. 1-5.

\bibitem{FemtoAccess2}
S. Joshi, R. C.C. Cheung, P. Monajemi, and J. Villasenor, ``Traffic-based study of femtocell access policy
impacts on HSPA service quality,'' in {\it Proc. IEEE GLOBECOM}, Dec. 2009, pp. 1-6.

\bibitem{FemtoAccess3}
P. Xia, V. Chandrasekhar, and J. G. Andrews, ``Open vs. Closed Access Femtocells in the Uplink,'' {\it Under revision IEEE Trans. Wireless Commun.}, February 2010, [Online] Available at http://arxiv.org/abs/1002.2964

\bibitem{FemtoAccess4}
D. Choi, P. Monajemi, S. Kang, and J. Villasenor, ``Dealing with loud neighbors: the benefits and tradeoffs of adaptive femtocell
access,'' in {\it Proc. IEEE GLOBECOM}, Dec. 2008, pp. 1-5.

\bibitem{FemtoAccess5}
G. de la Roche, A. Valcarce, D. L$\acute{\mathrm{o}}$pez-P$\acute{\mathrm{e}}$rez, and J. Zhang, ``Access control mechanisms for femtocells,'' {\it IEEE Commun. Mag.},
vol. 48, no. 1, pp. 33-39, Jan. 2010.

\bibitem{FemtoAccess6}
A. Valcarce, D. L$\acute{\mathrm{o}}$pez-P$\acute{\mathrm{e}}$rez, G. de la Roche, and J. Zhang, ``Limited access to OFDMA femtocells,'' in {\it Proc. IEEE PIMRC}, Sep. 2009, pp. 1-5.

\bibitem{FemtoAccess7}
D. L$\acute{\mathrm{o}}$pez-P$\acute{\mathrm{e}}$rez, A. Valcarce, G. de la Roche, and J. Zhang, ``Access methods to WiMAX Femtocells: A downlink system-level case study,'' in {\it Proc. IEEE ICCS}, Nov. 2008, pp. 1657-1662.

\bibitem{FemtoUL}
V. Chandrasekhar and J. G. Andrews, ``Uplink capacity and interference avoidance for two-tier femtocell networks,'' in {\it IEEE Trans.
Wireless Commun.}, vol. 8, no. 7, pp. 3498-3509, July 2009.


\bibitem{PPP2}
M. Haenggi, J. G. Andrews, F. Baccelli, O. Dousse, and
M. Franceschetti, ``Stochastic geometry and random graphs for
the analysis and design of wireless networks,'' {\it IEEE J.
Select. Areas Commun.}, vol. 27, no. 7, pp. 1029-1046, Sep. 2009.

\bibitem{AndGan10}
J. G. Andrews, R. K. Ganti, N. Jindal, M. Haenggi, and S. Weber, ``A Primer on Spatial Modeling and Analysis in Wireless Networks,'' to appear, {\it IEEE Communications Magazine}, Nov. 2010.

\bibitem{Baccelli}
F. Baccelli, B. Blaszczyszyn, and P. Muhlethaler, ``An ALOHA protocol
for multihop mobile wireless networks,'' {\it IEEE Trans. Inform. Theory},
vol. 52, no. 2, pp. 421-436, Feb. 2006.

%\bibitem{Baccelli1}
%F. Baccelli, B. Blaszczyszyn, and F. Tournois, ``Spatial averages of
%coverage characteristics in large CDMA networks,'' {\it Wireless Networks},
%vol. 8, no. 6, pp. 569-586, Nov. 2002.
%
%\bibitem{PPP2}
%C. C. Chan and S. Hanly, ``Calculating the outage probability in a CDMA
%network with spatial Poisson traffic,'' {\it IEEE Trans. Veh. Technol.}, vol. 50, no. 1, pp. 183-204, Jan. 2001.
\bibitem{ASE}
M. S. Alouini and A. J. Goldsmith, ``Area spectral efficiency of cellular
mobile radio systems,'' {\it IEEE Trans. Veh. Technol.}, vol. 48, no. 4, pp.
1047-1066, July 1999.

\bibitem{3GPPTS}
3GPP, ``Base Station radio transmission and reception (FDD),'' in 3GPP TS
25.104 V8.10.0, 2010.

\bibitem{HyperGeo}
A. Jeffrey, and D. Zwillinger, {\it Tables of
integrals, series, and products}. Academic Press, 2007.

\bibitem{Integral2}
http://functions.wolfram.com/ElementaryFunctions/Exp/21/01/02/02/02/
\end{thebibliography}
\end{document}